\title{Polynomial-time homology for simplicial Eilenberg--MacLane spaces
}

\newif\ifcmts
\cmtstrue

\documentclass[11pt,a4paper]{article}
\usepackage{a4}
\usepackage{ifthen}
\newboolean{MareksFormat}
\InputIfFileExists{MareksSwitch}{}{}
\ifthenelse{\boolean{MareksFormat}}{
 \usepackage{geometry}
 \geometry{papersize={145mm,180mm},top=2.5mm,bottom=2.5mm,left=1.5mm,right=1.5mm}
}{}

\usepackage{graphicx}
\usepackage{amsmath}
\usepackage{amsthm}
\usepackage{bbm}
\usepackage{url}

\theoremstyle{plain}
\newtheorem{theorem}{Theorem}[section]
\newtheorem{lemma}[theorem]{Lemma}
\newtheorem{corol}[theorem]{Corollary}

\theoremstyle{definition}

\DeclareMathOperator{\im}{im} 

\renewcommand\:{\colon}

\newcommand{\alterdef}[1]{\!\left\{\!\!\begin{array}{ll}
                                   #1 \end{array}  \right. }

\newcommand\compos{}  

\newcommand{\Z}{\mathbbm{Z}}
\newcommand{\N}{\mathbbm{N}}

\newcommand{\R}{\mathbbm{R}}

\renewcommand{\SS}{\mathcal{S}}

\newcommand{\CC}{\mathcal{C}}
\newcommand{\TT}{\mathcal{T}}
\newcommand{\II}{\mathcal{I}}

\newcommand\ttau{{\tilde\tau}}
\newcommand\tsigma{{\tilde\sigma}}
\newcommand\thedim{k}
\newcommand\diff{d}
\newcommand\dtC{\raisebox{0.2ex}{$\stackrel{\raisebox{-0.3ex}[0ex][-0.8ex]{$\scriptscriptstyle\approx$}}C$}}

\DeclareMathOperator{\id}{id}   
\DeclareMathOperator{\size}{{\sf size}}
\DeclareMathOperator{\reach}{{\sf reach}}
\DeclareMathOperator{\treach}{{\sf t-reach}}
\DeclareMathOperator{\lpow}{{\sf lpow}}
\DeclareMathOperator{\ltrim}{{\sf ltrim}}
\DeclareMathOperator{\ltrims}{\mbox{${\sf ltrim{*}}$}}
\DeclareMathOperator{\ones}{{\sf ones}}
\DeclareMathOperator{\Bas}{Bas}
\newcommand\ndg{{\rm ndg}}
\newcommand\Vbc{V_{\rm bch}}
\newcommand\Vbs{V_{\rm bs}}
\DeclareMathOperator{\LAS}{{\sf LAS}}
\newcommand\EC{EC}
\newcommand\xample[2]{Example: $\tau=[#1]$, $\tau'=[#2]$.}

\def\caseA{A}
\def\caseB{B}
\def\caseC{C}
\def\caseD{D}
\def\caseE{E}
\def\caseEa{F}
\def\caseF{G}
\def\caseG{H}
\def\caseH{I}
\def\casedA{dA}
\def\casedB{dB}
\def\casedC{dC}
\def\casedD{dD}
\def\casedH{dI}

\long\def\onefigure#1#2{
\begin{figure*}[tbp]
\begin{center}
#1
\end{center}
\caption{#2}
\end{figure*}
}
\def\immediateFigure#1{%
\smallskip\begin{center}#1\end{center}\smallskip }
\newcommand{\labfig}[2]  
{\onefigure{\mbox{\includegraphics{#1}}}{\label{f:#1} #2} }

\newcommand{\labfigw}[3]  
{\onefigure{\mbox{\includegraphics[width=#2]{#1}}}{\label{f:#1} #3}}

\newcommand{\immfig}[1]  
{\immediateFigure{\mbox{\includegraphics{#1}}}}

\newcommand{\immfigw}[2] 
{\immediateFigure{\mbox{\includegraphics[width=#2]{#1}}}}

\def\indef#1{\emph{#1}}

\newcommand{\heading}[1]{\vspace{1ex}\par\noindent{\bf\boldmath #1}}

\ifcmts
\newcommand{\marrow}{\marginpar{\boldmath$\longleftarrow$}}
\newcommand{\jirka}[1]{\ifhmode\newline\fi\marrow \textsf{*** (JIRKA: ) #1\newline}}
\newcommand{\marek}[1]{\ifhmode\newline\fi\marrow \textsf{*** (MAREK: ) #1\newline}}

\else
\newcommand{\marrow}{}
\newcommand{\jirka}[1]{}
\newcommand{\marek}[1]{}

\fi

\def\kamsymb{{\rm a}}
\def\ethsymb{{\rm b}}
\def\grasymb{{\rm *}}
\def\grensymb{{\rm c}}

\author{
Marek Kr\v{c}\'al$^{\kamsymb}$ \and
Ji\v{r}\'{\i} Matou\v{s}ek$^{\kamsymb, \ethsymb}$
\and Francis Sergeraert$^\grensymb$
}

\begin{document}

\maketitle

{\renewcommand\thefootnote{\grasymb}
\footnotetext{
The research by J.\,M. and M.\,K.\ was supported
by the Institute for Theoretical Computer Science (ITI),
Charles University, Prague (project 1M0545
of the Czech Ministry of Education) and by the
ERC Advanced Grant No.~267165.
The research by M.\,K. was also supported by the project GAUK~49209.
 }
}


{\renewcommand\thefootnote{\kamsymb}
\footnotetext{Department of Applied Mathematics,
Charles University, Malostransk\'{e} n\'{a}m.~25,
118~00~~Praha~1,  Czech Republic}
}
{\renewcommand\thefootnote{\ethsymb}
\footnotetext{Institute of  Theoretical Computer Science,
ETH Zurich, 8092~Zurich, Switzerland}
}
{\renewcommand\thefootnote{\grensymb}
\footnotetext{Institut Fourier,
BP~74,
38402~St Martin, d'H\`eres Cedex,
France}}


\begin{abstract}
In an earlier paper of \v{C}adek, Vok\v{r}\'{\i}nek, Wagner,
and the present authors, we investigated an algorithmic problem
in computational algebraic topology, namely, the
computation of all possible homotopy classes of maps between
two topological spaces, under suitable restriction on the spaces.

We aim at showing that, if the dimensions of the considered spaces
are bounded by a constant, then the computations can be done in
polynomial time. In this paper we make a significant technical
step towards this goal: we show that the Eilenberg--MacLane
space $K(\Z,1)$, represented as a simplicial group, can be
equipped with \emph{polynomial-time homology} (this is a polynomial-time
version of \emph{effective homology} considered in previous
works of the third author and co-workers).

To this end, we construct a suitable \emph{discrete vector field},
in the sense of Forman's \emph{discrete Morse theory}, on $K(\Z,1)$.
The construction is purely combinatorial and it can be understood
as a certain procedure for reducing finite sequences of integers,
without any reference to topology.

The Eilenberg--MacLane spaces
are the basic building blocks in a \emph{Postnikov system},
which is a ``layered'' representation of a topological space
suitable for homotopy-theoretic computations.
Employing the result of this paper
together with other results on polynomial-time homology,
in another paper we obtain, for every fixed $\thedim $,  a polynomial-time
algorithm for computing the $\thedim $th homotopy group $\pi_\thedim(X)$
of a given simply connected space $X$, as well as the first $\thedim $
stages of a Postnikov system for~$X$, and also
a polynomial-time version  of the algorithm
of \v{C}adek et~al.\ mentioned above.
\end{abstract}

\section{Introduction}

Recently our co-authors and we \cite{CKMSVW11} have developed an algorithm
for a problem in computational algebraic topology (more precisely,
in computational homotopy theory),
 namely, computing all homotopy classes of maps between two topological spaces
$X$ and $Y$ (given as finite simplicial complexes, say),
under certain natural conditions on $X$ and $Y$.

Our original motivation was understanding the computational
complexity of the $\Z_2$-index of a given
$\Z_2$-space, which is a quantity appearing in various applications
of topology in combinatorics and geometry (e.g., topological lower bounds
for the chromatic number of a graph, or an algorithm for testing
the embeddability of a given simplicial complex into $\R^d$).
We hope to reach results in this direction in the future,
and we also expect that the developed methods will be applicable
for other natural problems (such as extendability of maps;
as a concrete application, it was already possible to answer a question
of Franek et al.~\cite{Franek-al} on testing nullhomotopy of maps
into a sphere). For more information on this project we refer to 
\cite{CKMSVW11,polypost,ext-hard}.

\heading{Towards polynomial-time homology. }
The implementation of some of the operations in the algorithm
of \cite{CKMSVW11} relies on the methods of \emph{effective
homology}, initiated by the third author
in \cite{Sergeraert:ComputabilityProblemAlgebraicTopology-1994}
and further developed by him and his co-workers
(see, e.g., \cite{RubioSergeraert:ConstructiveAlgebraicTopology-2002,RomeroRubioSergeraert,SergerGenova}). These provide algorithmic solutions
of many problems in algebraic topology, but so far no analysis
of their running time was available, and for some parts the running
time can actually be exponential.

One of our aims is to obtain \emph{polynomial-time} algorithms for
these tasks where possible, or alternatively, show computational
hardness.

Let us stress that by ``polynomial-time'' we mean, throughout this paper,
\emph{polynomial-time for every fixed dimension}.
Thus, assuming that the input to an algorithm
is a space  represented as a finite simplicial complex $X$,
we want that the running time is polynomial in the number
of simplices of $X$, but the polynomial may depend on the dimension $\thedim$ of~$X$
(and the dependence on $\thedim $ may be exponential or even worse).
Of course, one could be even more ambitious and ask for a polynomial
dependence on $\thedim $ as well; however, we do not expect such algorithms
to exist, in view of computational hardness results
\cite{Anick-homotopyhard,ext-hard}.

To integrate this effort with existing algorithms, we start
with the framework of effective homology mentioned above, and
we introduce an analogous definition of \emph{polynomial-time
homology}; see Section~\ref{s:eh-intro}. In another paper
\cite{polypost}, we show
that  various known constructions and operations on objects
with effective homology have polynomial-time versions.
With a repertoire of such operations, we also obtain a
polynomial-time version of the algorithm of
\cite{CKMSVW11}, as well as other algorithms, such as
computing the higher homotopy group $\pi_\thedim(X)$ in
polynomial time for every fixed $\thedim $, or computing the
first $\thedim $ stages of a Postnikov system for~$X$.

\heading{This paper. } Here we make a significant step in this
development.
First we set up the framework of polynomial-time homology
(modeled after effective homology mentioned above)
and some tools of general applicability.
Then, in the second part of the paper, we present
our main technical result. The problem which we solve
can be formulated purely combinatorially, although in this language it perhaps
doesn't sound extremely natural:
it is a question about
reducing finite sequences of integers by certain simple operations.
We will state it below, and no topological notion at all is required
for understanding this problem and our solution.

However, to explain its role in computational topology,
we first need to sketch some background information.
A standard reference for this material is
May \cite{May:SimplicialObjects-1992}; a concise overview is given
in \cite{CKMSVW11}, and more leisurely explanations can be found in
\cite{SergerTrieste} or~\cite{SergerGenova}.

A common technique in mathematics and in computer science is to
decompose a general, presumably complicated object into
simpler building blocks. For the purposes of understanding
continuous mappings going \emph{into} a given topological
space $Y$, a suitable decomposition is a
\emph{Postnikov system} for~$Y$; indeed, this is a crucial
ingredient of the algorithm in~\cite{CKMSVW11}.

We do not need to define the rather complicated notion
of Postnikov system here; it suffices to say
that its ``building blocks'' belong to a particular class
of topological spaces, called \emph{Eilenberg--MacLane spaces}
and denoted by $K(G,\thedim)$, where $G$ is an Abelian group
 and $\thedim \ge 1$ is an integer.
In the Eilenberg--MacLane spaces appearing in a Postnikov system
for $Y$, the role of the group $G$ is played by the
homotopy groups $\pi_i(Y)$, $i \ge 2$.

In topology, $K(G,\thedim)$ is defined as
a topological space $T$ whose homotopy groups
satisfy $\pi_\thedim(T)\cong G$ and $\pi_i(T)=0$ for all $i \ne \thedim$.
It is determined uniquely up to homotopy equivalence
(in the class of all CW complexes).

Generally speaking, the spaces $K(G,\thedim)$  are
infinite-dimensional and they do not look like very simple
objects (with the exception of $K(\Z,1)$, which is homotopy
equivalent to the circle $S^1$). However, they are in some
sense the simplest possible spaces concerning maps going into
them. These spaces are of basic importance in algebraic
topology, and a lot of work has been devoted to studying
their properties, and in particular, computing their homology
and cohomology (Serre \cite{serre53} and H.~Cartan
\cite{cartan} are two of the most famous classical works;
see, e.g.,
 Cl\'ement \cite{clement} for an overview and some computational
aspects). We also refer to Romero and Rubio \cite{RomeroRubioKG1}
for an algorithmic study of $K(G,1)$ for noncommutative groups~$G$.

For the intended algorithmic use, we need a particular representation
of $K(G,\thedim)$; namely, we need it represented as a particular kind
of a simplicial set (simplicial sets will be briefly
introduced in Section~\ref{s:eh-intro} below),
a so-called \emph{Kan simplicial set}. We
use the standard Eilenberg--MacLane simplicial model for \(K(G,\thedim)\);
see
\cite[Chapter III]{EilenbergMacLane:GroupsHPin1-1953},
\cite[Chapter V]{May:SimplicialObjects-1992}.

For the algorithms, we need to equip the simplicial Eilenberg--MacLane
spaces with polynomial-time homology. The $K(G,\thedim)$ we may encounter
can have any finitely generated Abelian group as $G$, and any
positive integer as~$\thedim$.

However, in this paper we will deal only with $K(\Z,1)$,
which serves as a base case, while the other $K(G,\thedim)$
can be obtained from it using several operations. First, for
direct products of groups, we have $K(G\times H,\thedim)\cong
K(G,\thedim)\times K(H,\thedim)$, and so, with a general
product operation available, we may assume that $G$ is
cyclic. Second, a general construction, known as the
\emph{classifying space} (actually, in the simplicial
setting, we deal with the so-called $\overline W$-construction),
allows one to pass from
$K(G,\thedim)$ to $K(G,\thedim+1)$, so indeed $\thedim =1$ is
the important base case. Finally, polynomial-time homology
for $K(\Z/m\Z,1)$ can be obtained from that for $K(\Z,1)$
using another operation, namely, computing the \emph{base
space of a fibration}. These reductions are discussed in
\cite{SergerGenova}, and
polynomial-time versions are discussed in \cite{polypost};
here we just wanted to
provide a quick explanation of why the $K(\Z,1)$ case
deserves special attention.\footnote{Curiously, $K(\Z,1)$ as
a topological space almost can't be simpler---as we
mentioned, it is homotopy equivalent to the circle $S^1$, and other
Eilenberg--MacLane spaces are much more complicated. But we
need to work with the Kan simplicial model of $K(\Z,1)$ as introduced
above, which has infinitely many simplices in every dimension $\thedim \ge 1$.
As we will see, for effective (or polynomial-time) homology,
it is not sufficient to know, for example, that
$H_2(K(\Z,1))=0$, but we need to be able to actually compute
``witnesses'' for it; that is, given a 2-cycle $z_2$ on
$K(\Z,1)$, compute a 3-chain for which $z_2$ is its
boundary. This problem would be trivial for the standard simplicial
representation of $S^1$ with one vertex  and one edge,
but it is not trivial for the considered
Kan model of $K(\Z,1)$.}

\heading{The combinatorial problem about integer sequences. }
The $\thedim $-dimensional simplices of the standard simplicial
model of $K(\Z,1)$, $\thedim =0,1,
\ldots$  can be represented by $\thedim $-term sequences of integers.
With the traditional ``bar notation'', such a sequence is written as
\begin{equation}\label{e:theseq}
\sigma=[a_1\,|\,a_2\,|\cdots|\,a_\thedim],\ \ a_1,a_2,\ldots,a_\thedim\in\Z.
\end{equation}
In the rest of this introduction, a ``$\thedim $-dimensional simplex''
will thus be synonymous with a ``$\thedim $-term sequence of integers''.

For our problem we consider only \emph{nondegenerate
simplices}, represented by sequences with \emph{no zero
terms}. Thus, from now on, we always assume that all the
$a_i$ are nonzero.

For each $\thedim $, there are $\thedim +1$ \emph{face operators}
$\partial_0,\partial_1,\ldots,\partial_\thedim$, which
map $\thedim $-term sequences to $\thedim -1$ term sequences:
$\partial_0$ deletes the first component, $\partial_\thedim$ deletes
the last component, and for $i =1,2,\ldots,\thedim-1$, $\partial_i$
reduces the number of components by one by adding together the
$i $th and $(i+1)$st component. More formally, with $\sigma$ as above,
$$
\partial_0\sigma=[a_2\,|\cdots|\,a_\thedim],\ \
\partial_\thedim\sigma=[a_1\,|\cdots|\,a_{\thedim-1}],
$$
$$
\partial_i\sigma=[a_1\,|\cdots|\,a_{i-1}\,|\,
a_{i}+a_{i+1}\,|\,a_{i+2}\,|\cdots|\,a_\thedim],\ \ \ 1\le i\le \thedim-1.
$$

The goal is to divide the set of all possible finite
sequences $\sigma$ of nonzero integers into three classes
$\SS$, $\TT$, and $\CC$ (the \emph{source simplices},
\emph{target simplices}, and \emph{critical simplices}),
and construct a bijection $V\:\SS\to\TT$ (which will be called
a \emph{discrete vector field}), such that
for every $\sigma\in\SS$, we have $\sigma=\partial_i V(\sigma)$
for exactly one $i $. We also require certain additional properties, which
we explain next.

With $\SS,\TT,\CC$, and $V$ as above,
 let us consider a sequence (simplex) $\tilde\sigma\in\SS$ of
some dimension $\thedim $, and let us say that a simplex $\tau$
(of dimension $\thedim $ or $\thedim+1$)
is \emph{reachable} from $\tilde\sigma$ if it can be reached from $\tilde\sigma$
by finitely many moves, where the allowed moves are
\begin{itemize}
\item passing from a current simplex $\sigma\in\SS$ to the simplex
$\tau=V(\sigma)\in\TT$, and
\item passing from a current simplex $\tau\in\TT$ to a simplex
$\sigma=\partial_i\tau\in\SS\cup\CC$
such that $\tau\ne V(\sigma)$, where $i \in\{0,1,\ldots,\thedim\}$.
\end{itemize}

With these definitions, it is required that
\begin{enumerate}
\item[(i)] for every $\thedim $, $\CC$ contains only finitely many
$\thedim $-dimensional simplices; and
\item[(ii)] starting with any $\tilde\sigma$, we can never
make an infinite sequence of allowed moves; that is,
we can reach only finitely many simplices, and we also
cannot get into a cycle.
\end{enumerate}
Moreover, we measure the \emph{size} of a simplex
$\sigma=[a_1|\cdots|a_\thedim]$ as the total number
of bits needed to write down $a_1,\ldots,a_\thedim$; more formally,
we set $\size(\sigma):=\sum_{i=1}^\thedim\size(a_i)$
and $\size(a):=1+\lfloor\log_2(|a|+1)\rfloor$.
Then we also require that
\begin{enumerate}
\item[(iii)] For every $\thedim $-dimensional simplex $\tilde\sigma$,
the sum of $\size(\sigma)$ over all $\sigma$ reachable from
$\tilde\sigma$ is bounded by a polynomial (depending on $\thedim $)
in $\size(\tilde\sigma)$.
\end{enumerate}

To illustrate these definitions, let us present
a classical vector field $V_{\rm EML}$
 due to Eilenberg and Mac Lane, which satisfies
(i) and (ii) (and yields effective homology for $K(\Z,1)$) but not~(iii).

There are only two critical
simplices, the 0-dimensional  $[\,]$ (the empty
sequence)
and the 1-dimensional $[1]$.\footnote{This actually
 corresponds to the topological
fact that the considered $K(\Z,1)$, as a topological space,
 is homotopy equivalent to $S^1$; $[\,]$ represents a vertex,
and $[1]$ an edge glued to that vertex by both ends, forming an $S^1$.}
The set $\SS$ of source simplices consists of the sequences with $a_1\ne 1$,
while $\TT$ contains the sequences with $a_1=1$ (the
two critical simplices are exceptions to this rule).

For $\sigma=[a_1|\cdots|a_\thedim]\in\SS$, $a_1\ne 1$, the vector
field $V_{\rm EML}$ is defined by
$$
V_{\rm EML}(\sigma):=\alterdef{{[1|a_1-1|a_2|\cdots|a_\thedim]}&\mbox{for }a_1>1,\\
                     {[1|a_1|a_2|\cdots|a_\thedim]}&\mbox{for }a_1<0.}
$$
It can be checked that, for any starting $\tilde\sigma$, the
sequence of moves is determined uniquely (there is no branching).

It is easy to see that, for
a positive integer $a$, the sequence of moves starting from
$[a]$ is $[a]\to[1|a-1]\to[a-1]\to[1|a-2]\to[a-2]\to\ldots$;
there are about $a$ moves, and this is \emph{exponential}
in the number of bits of~$a$. Thus, condition (iii) above
indeed fails.

We will provide a solution satisfying (i)--(iii) in Section~\ref{s:KZ1}.
Before that, we introduce simplicial sets, polynomial-time homology,
and discrete vector fields in general.

\section{Simplicial sets with polynomial-time homology}\label{s:eh-intro}

\heading{Simplicial sets. } A simplicial complex is a way of specifying
a topological space in purely combinatorial terms,
and also a way of presenting a topological space as an input to an algorithm;
we assume that the reader is basically familiar with this concept.

A simplicial set can be regarded as a generalization of a simplicial
complex;  it is more complicated, but more powerful
and flexible. The algorithms we consider use simplicial sets as the
main data type for representing topological
spaces and their maps. A friendly introduction
to simplicial sets is \cite{Friedm08}, and another introductory
treatment can be found in \cite{SergerTrieste}; older compact
sources are, e.g.,
\cite{Curtis:SimplicialHomotopyTheory-1971,May:SimplicialObjects-1992},
 and \cite{GoerssJardine} is a more modern and comprehensive treatment.

Similar to a simplicial complex,
a simplicial set is a space built of vertices, edges, triangles,
and higher-dimensional simplices, but simplices are allowed to be glued
to each other and to themselves in more general ways. For example,
one may have several 1-dimensional simplices connecting the same
pair of vertices, a 1-simplex forming a loop,
two edges of a  2-simplex identified
to create a cone, or the boundary of a 2-simplex all contracted
to a single vertex, forming an $S^2$.
\immfig{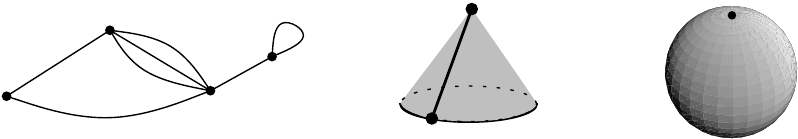}

Another new feature of a simplicial set, in comparison with a simplicial
complex, is the presence of \emph{degenerate simplices}. For example,
the edges of the triangle with a contracted boundary (in the last
example above) do not disappear, but each of them becomes
a degenerate 1-simplex.

A simplicial set $X$ is represented as a sequence
$(X_0,X_1,X_2,\ldots)$ of mutually disjoint sets, where the
elements of $X_\thedim$ are called the \emph{$\thedim
$-simplices of $X$} (we note that, unlike for simplicial
complexes, a simplex in a simplicial set need not be
determined by the set of its vertices; indeed, there can be
many simplices with the same vertex set). For every $\thedim
\ge 1$, there are $\thedim +1$ mappings
$\partial_0,\ldots,\partial_\thedim\:X_\thedim\to
X_{\thedim-1}$ called \indef{face operators}; the intuitive
meaning is that for a simplex $\sigma\in X_\thedim$,
$\partial_i\sigma$ is the face of $\sigma$ opposite to the
$i$th vertex. Moreover, there are $\thedim +1$ mappings
$s_0,\ldots,s_\thedim\:X_\thedim\to X_{\thedim+1}$ (opposite
direction) called the \emph{degeneracy operators}; the
approximate meaning of $s_i\sigma$ is the degenerate simplex
which is geometrically identical to $\sigma$, but with the
$i$th vertex duplicated. A simplex is called
\indef{degenerate} if it lies in the image of some $s_i$;
otherwise, it is \indef{nondegenerate}. We write $X^\ndg$ for
the set of all nondegenerate simplices of~$X$.

There are natural axioms that the $\partial_i$
and the $s_i$ have to satisfy, but we will not list them here,
since we won't really use them. Moreover, the usual definition
of simplicial sets uses the language of category theory
and is very elegant and concise; see, e.g.,
\cite[Section I.1]{GoerssJardine}.

Every simplicial set $X$ specifies a topological space $|X|$,
the \emph{geometric realization} of $X$. It is obtained by assigning
a geometric $\thedim $-dimensional simplex to each nondegenerate
$\thedim $-simplex of $X$,
and then gluing these simplices together according to the
face operators; we refer to the literature for the precise definition.

There is a canonical way of converting a simplicial complex
to a simplicial set; basically, one just needs to add
appropriate degenerate simplices.

We have already given a relatively sophisticated example of a
simplicial set, namely, $K(\Z,1)$, or more precisely, the
standard Eilenberg--MacLane representation of $K(\Z,1)$ as a
Kan simplicial set\footnote{We won't define a Kan simplicial
set, but we just mention a key property, which is the reason
why these simplicial sets are essential to the considered
algorithms. Namely, if $X$ is a simplicial set and $Y$ is a
Kan simplicial set, then every \emph{continuous} map
$|X|\to|Y|$ is homotopic to a \emph{simplicial} map $X\to Y$.
Thus, continuous maps into $Y$ have a combinatorial
representation, describing them up to homotopy.}
 as defined in the introduction
(except that we haven't yet specified the degeneracy
operators, which are very simple: $s_i$
inserts $0$ after the $i $th component of a sequence).

\heading{Representing infinite simplicial sets. }
In many areas where computer scientists seek efficient
algorithms, both the input objects and intermediate
results in the algorithms are finite, and they can be
explicitly represented in the computer memory; this is
the case, e.g., for algorithms dealing with graphs or with
matrices.

In contrast, in the algorithms for homotopy-theoretic
questions considered here and in related
works, we need to deal with infinite objects.
For example, even if the input is a finite simplicial complex,
its Postnikov system (mentioned in the introduction)
is made of Eilenberg--MacLane spaces, such as $K(\Z,1)$,
represented as Kan simplicial sets, and these are
necessarily infinite.
More concretely, as we have seen, $K(\Z,1)$ has infinitely
many simplices in each dimension $\thedim \ge 1$, and thus we cannot
explicitly store even the part up to some fixed dimension.

For algorithmic purposes, we thus represent a simplicial set $X$
by  a collection of several algorithms,
which allow us to access certain information about $X$, without having
all of it explicitly stored in memory. (In computer science,
this is also called a \emph{black box} or \emph{oracle} representation
of $X$, and in the terminology of object-oriented programming,
we can think of $X$ as an instance of a \emph{class} ``simplicial set''.)
A similar representation is used for other kinds of infinite
topological or algebraic objects as well.

\heading{Locally effective simplicial sets.
}
For some computations, it may be sufficient to represent $X$
by a black box providing only ``local'' information about $X$,
and in that case, in accordance with the terminology in earlier papers, e.g.,
\cite{RubioSergeraert:ConstructiveAlgebraicTopology-2002,SergerGenova,RoSe-vecfields}, we speak of a \emph{locally effective representation}.

Concretely, let $X$ be a simplicial set, and suppose
that some computer representation (``encoding'') for the simplices
of $X$ has been fixed. For example,  in the case of $K(\Z,1)$,
we can fix the representation of the simplices of $K(\Z,1)$ by
integer sequences, and represent the integers in the sequences
by the standard binary encoding. We say that $X$ is
a \emph{locally effective simplicial set} if algorithms are
available that, given (an encoding of) a $\thedim $-simplex $\sigma$
of $X$ and $i \in\{0,1,\ldots,\thedim\}$, computes the simplex
$\partial_i\sigma$, and similarly for the degeneracy
operators~$s_i$. Briefly speaking, the face and degeneracy
operators should be computable maps.


\heading{Computing global information. } Suppose that we want to
compute some ``global'' information about a given simplicial
set $X$, for example, the $\thedim $th homology group $H_\thedim(X)$.
Then a locally effective representation of $X$ is typically
insufficient, and we need to augment it in some way.

Of course, in the particular example with the homology groups, we could
insist that $X$ be augmented with a black box that, given $\thedim $,
returns some representation of $H_\thedim(X)$. The problem is that
$X$ may not be given to us directly; rather, we may need to construct
it from other simplicial sets by a sequence of various operations.
For example, in the introduction we mentioned that the Eilenberg--MacLane
spaces $K(G,\thedim)$ can be constructed starting with $K(\Z,1)$ and
applying operations of several kinds, such as product or
classifying space.\footnote{As another, perhaps more sophisticated
example, we can mention the computation of the homotopy group
$\pi_\thedim(X)$ for a $1$-connected simplicial set $X$:
for this, given $X$, one first produces another
simplicial set $X'$ from $X$, by a sequence of operations
that ``kill'' the first $\thedim-1$ homotopy groups, and then
$\pi_\thedim(X)$ is computed as $H_\thedim(X')$ using the \emph{Hurewicz
isomorphism}.} Then, for example, a black box for computing
the homology groups of $X$ is not in itself sufficient to
compute the homology groups of the classifying space
of~$X$.

The third author and his co-authors have developed a more sophisticated
way of augmenting a locally effective simplicial set $X$ with
homological information, which is captured in the notion
of a \emph{simplicial set with effective homology}. These
simplicial sets do possess a black box for computing homology groups,
but they are also equipped with additional information, which
makes them stable under a large repertoire of operations:
if we apply some of the ``classical'' operations, such as
product, classifying space, loop space, etc. to simplicial
sets with effective homology, the result is again a simplicial
set with effective homology (and in particular, it has a black
box for computing homology groups).\footnote{One can also consider other kinds of objects with effective
homology, such as chain complexes, but for concreteness, we will
stick to simplicial sets.}

It may be useful to keep in mind that, since a simplicial set
is represented by a black box, operations on such simplicial
sets are performed by \emph{composition of algorithms}; i.e.,
the black box for the new simplicial set operates by calling
the black boxes of the old sets and processing the values
returned by them.\footnote{This feature makes it very natural
to implement algorithms from this area using
\emph{functional programming languages}, as was done
for the package \emph{Kenzo}; see, e.g., \cite{fKenzo}.}

For defining a simplicial set with effective homology,
and their polynomial-time counterpart, we need to
recall some notions concerning chain complexes.

\heading{Chain complexes. }  For our purposes, a \emph{chain
complex} $C_*$ is a sequence
$(C_\thedim)_{\thedim=-\infty}^\infty$ of \emph{free}
$\Z$-modules (i.e., free Abelian groups), together with a
sequence $(\diff_\thedim\:C_\thedim\to
C_{\thedim-1})_{\thedim=-\infty}^\infty$ of group
homomorphisms.\footnote{These chain complexes are over $\Z$;
more generally, one considers chain complexes over a
commutative ring $R$, where the $C_\thedim$ are $R$-modules.
These are needed, among others, for homology with
coefficients in $R$. But for our purposes, homology with
integer coefficients suffices; if needed, homology groups
with other coefficients can be computed using universal
coefficient theorems. Alternatively, all of the theory can be
built with coefficients from a fixed ring $R$, provided that
$R$ is equipped with sufficiently strong algorithmic
primitives.} The $C_\thedim$ are the \emph{chain groups},
their elements are called \emph{$\thedim$-chains}, and the
$\diff_\thedim$ the \emph{differentials}.
The differentials have to satisfy $\diff_{\thedim-1}\compos \diff_\thedim=0$
for every~$\thedim $ (here $\diff_{\thedim-1}\compos \diff_\thedim$
denotes the composition of maps).
We also recall that the $\thedim $th homology group $H_\thedim(C_*)$
of the chain complex $C_*$ is defined as the
factor-group $\ker \diff_\thedim/\im\diff_{\thedim+1}$.

For every simplicial set $X$, there is a canonically associated
chain complex, which is used to define the homology groups
$H_\thedim(X)$. Actually, there are two natural possibilities,
depending on whether degenerate simplices are taken into account.
We use the \emph{normalized} chain complex, which
is based solely on the nondegenerate simplices.
We reserve the simple notation $C_*(X)$ for it.

Thus, $C_\thedim(X)$ denotes the free Abelian group over
$X_\thedim^\ndg$, the set of all $\thedim$-dimensional
nondegenerate simplices (in particular, $C_\thedim(X)=0$
for $\thedim <0$). This means that a $\thedim $-chain is a
formal sum
$$
c=\sum_{\sigma\in X_\thedim^\ndg} \alpha_\sigma \cdot\sigma,
$$
where the $\alpha_\sigma$ are integers, only finitely many of them
nonzero.
 The differentials are defined in a standard way
using the face operators: for $\thedim $-chains of the form $1\cdot\sigma$,
which constitute a basis of $C_\thedim(X)$, we set
$\diff_\thedim (1\cdot\sigma):=\sum_{i=0}^\thedim(-1)^i\cdot\partial_i\sigma$
(some of the $\partial_i\sigma$ may be degenerate simplices; then
they are ignored in the sum),
and this extends to a homomorphism in a unique way (``linearly'').

We note that if $X$ is a locally effective simplicial set,
then the $\thedim $-chains of $C_*(X)$ are finite objects;
a $\thedim $-chain $c$ can be represented by a list of the $\thedim $-simplices $\sigma$
on which $c$ is nonzero, and of the corresponding coefficients~$\alpha_\sigma$.
Then the differentials are computable maps.

However, if $X_\thedim^\ndg$ is infinite, then $C_\thedim(X)$ has
infinite rank, and we cannot use it directly for computing
homology groups. The solution adopted in effective homology
is to have, together with a locally effective simplicial
set $X$, a \emph{reduction} from $C_*(X)$ to
an ``effective'' chain complex $\EC_*$, for which each chain
group $\EC_\thedim$ has a finite rank.

\heading{Reductions. }  Let $C_*,\tilde C_*$ be two chain
complexes. To define a reduction from $C_*$ to $\tilde C_*$,
we first recall two other standard notions from homological
algebra: A \emph{chain map} $f\:C_*\to\tilde C_*$
is a sequence $(f_\thedim)_{\thedim=-\infty}^\infty$ of homomorphisms
$f_\thedim\:C_\thedim\to\tilde C_\thedim$ compatible with the differentials,
i.e., $f_{\thedim-1}\compos \diff_\thedim=\tilde\diff_{\thedim}\compos f_\thedim$.
If $f,g\:C_*\to\tilde C_*$ are two chain maps, then a
\emph{chain homotopy} of $f$ and $g$ is a sequence $(h_\thedim)_{\thedim=-\infty}^\infty$
of homomorphisms $h_\thedim\: C_\thedim\to\tilde C_{\thedim+1}$
such that $f-g=\tilde\diff_{\thedim+1}\compos h_\thedim +
h_{\thedim-1}\compos\diff_\thedim$.

Now a \emph{reduction} $\rho$ from $C_*$ to $\tilde C_*$ consists
of three maps $f,g,h$, such that
\begin{itemize}
\item $f\:C_*\to \tilde C_*$ and $g\:\tilde C_*\to C_*$
are chain maps;
\item the composition $f\compos g\:\tilde C_*\to\tilde C_*$
is equal to the identity $\id_{\tilde C_*}$, while the composition
$g\compos f\:C_*\to C_*$ is chain-homotopic to $\id_{C_*}$,
with $h\:C_*\to C_*$ providing the chain homotopy; and
\item $f\compos h=0$, $h\compos g=0$, and $h\compos h=0$.
\end{itemize}

The notion of reduction goes back to Eilenberg and Mac~Lane
\cite[Section~12]{EilenbergMacLane:GroupsHPin1-1953},
who called it  a \emph{contraction}.\footnote{They did not require
the condition $h\compos h=0$, but
 simple transformation converts a reduction without this
condition into another one satisfying it.}
It is routine to check that if there is a reduction from
$C_*$ to $\tilde C_*$, then $C_*$ and $\tilde C_*$ have isomorphic
homology groups in each dimension. Reductions can also be composed,
as follows:
if $(f,g,h)$ is a reduction from $C_*$ to $\tilde C_*$ and
$(f',g',h')$ is a reduction from $\tilde C_*$ to $\dtC_*$,
then $(f'\compos f,g\compos g', h+g\compos h'\compos f)$ is a reduction
from $C_*$ to~$\dtC_*$.

\heading{Effective homology. } We are getting close to
stating the definition of a simplicial set with effective
homology. The last step is to define what we mean by an
effective chain complex $\EC_*$. We assume that, first,
$\EC_*$ is \emph{locally effective}, meaning that each
chain group $\EC_\thedim$ has some distinguished basis $\Bas_\thedim$,
$\thedim$-chains are represented as linear combinations of
elements of $\Bas_\thedim$ (and thus they can be added and subtracted
algorithmically), and there is an algorithm
for evaluating the differentials $\diff_\thedim$.
Second, $\EC_*$ is \emph{effective}, which means, in addition to the
above, that there is an algorithm that, given $\thedim$,
outputs the list of elements of the distinguished basis $\Bas_k$;
in particular, this implies that each $\EC_\thedim$ has a finite
rank $r_\thedim$. We note that by combining the construction
of $\Bas_k$ with the ability to evaluate the differential $\diff_\thedim$,
we can compute the matrix of $\diff_\thedim$ with respect to
the distinguished bases $\Bas_\thedim$ and $\Bas_{\thedim-1}$.

We can now define a \emph{simplicial set with effective
homology} as a locally effective simplicial set $X$ together
with an effective chain complex $\EC_*$ and a reduction
$\rho$ from $C_*(X)$ to $\EC_*$, where the three maps
$f,g,h$ from the definition of reduction are
computable.\footnote{In \cite{SergerGenova} and in other
papers, effective homology is defined in a more general way,
using \emph{strong equivalence} of chain complexes instead of just a
reduction. A strong equivalence of $C_*$ and $\tilde C_*$ means
that there is an auxiliary chain complex $A_*$ and reductions
of $A_*$ to both $C_*$ and $\tilde C_*$. However, here the
simpler notion using a single reduction suffices, and this
only makes the result formally stronger, since a reduction is
a special case of a strong equivalence.}

In this paper we won't have the opportunity to demonstrate
the usefulness of effective homology in algorithms;
we refer to, e.g., \cite{SergerGenova,SergerTrieste,polypost}
for examples of applications. 
%

\heading{Polynomial-time homology. } The meaning of polynomial-time
homology for the simplicial set $K(\Z,1)$ considered in this
paper is defined in a straightforward way: we want the
face and degeneracy operators to be computable in polynomial time
(which is obvious in this particular case), and
$K(\Z,1)$ should be equipped
with effective homology as above in such a way that, for every $\thedim$,
the maps $f_\thedim,g_\thedim,h_\thedim$ are computable
in polynomial time, with the polynomial possibly depending
on $\thedim$ as usual.

We stress that since we deal with
a \emph{single} effective chain complex $\EC_*$, the
ranks $r_\thedim$ depend only on $\thedim$ and thus,
for $\thedim$ fixed, they are constants. The matrix
of the differential $\diff_\thedim$ in
$\EC_*$, too, is  a \emph{constant-size} object.

However, our setting with $K(\Z,1)$ is somewhat unusual in the analysis
of algorithms: We are dealing with a single simplicial set, fixed
once and for all, which does not depend on any input.
This is an exceptional setting; most algorithms work with objects
that do depend on the input. To draw an analogy from a different
area, the setting of the present paper can be compared to
seeking an algorithm for computing the $n$th digit of the number
$\pi$, while the more usual case would be to consider algorithms
for evaluating arithmetic expressions with arbitrary precision,
where we start with integer numbers as inputs and apply
addition, subtraction, multiplication, division, roots
and functions like $\exp$, $\ln$ or $\arcsin$.

To have an example from the area considered here, in an
algorithm for computing with a given topological space $X$,
say specified as a finite simplicial complex, we may need
polynomial-time homology for the Eilenberg--MacLane space
$K(\Z^n,1)$, where $n$ is a parameter depending on $X$. Then
we want that in the corresponding effective chain complex for
$K(\Z^n,1)$, the ranks $r_2,r_3$, etc. each depend
polynomially on~$n$. (Of course, for this to be useful, we
also need that $n$ depends at most polynomially on the size
of $X$.)


This example suggests that, in order to have a generally
useful notion of polynomial-time homology, we need to
define it formally for a whole family, typically infinite, of simplicial sets.
Here we present this issue briefly, referring to
\cite{polypost} for a more detailed discussion.

Let $\II$ be a set, typically countable, such that each element
$I\in\II$ has some agreed-upon computer representation (i.e. encoding
by a finite string of bits). A \emph{simplicial set parameterized
by $\II$} is a mapping $X$ that assigns a simplicial set $X(I)$
to each $I\in\II$. We also assume that the simplices of each $X(I)$
have some encoding by bit strings. Then we define a
\emph{locally polynomial-time simplicial set} as a
simplicial set $X$ parameterized by some $\II$ such that
the face and degeneracy operators on a $\thedim$-simplex
$\sigma$ of $X(I)$ can be evaluated in time polynomial
in $\size(I)+\size(\sigma)$, where the polynomial may depend
on $\thedim$ (and $\size(.)$ denotes the
the number of bits in the encoding).

Quite analogously, we define a chain complex $C_*=(C(I)_*:I\in\II)$
parameterized by a set $\II$. We say that such a $C_*$ is
\emph{locally polynomial-time} if each $C(I)_*$ is a locally
effective chain complex (and in particular, it has a distinguished
basis $\Bas(I)_\thedim$, and $\thedim$-chains are represented
w.r.t.~this basis), and  for each fixed $\thedim$,
the differential $(\diff_I)_\thedim$ on $C(I)_\thedim$
can be evaluated in time polynomial in $\size(I)$ plus
the size of the input $\thedim$-chain.  We note that addition
and subtraction of $\thedim$-chains are polynomial-time operations
automatically.

We say that a simplicial set $X$ parameterized
by a set $\II$ is  equipped with \emph{polynomial-time homology} if
the following hold.
\begin{itemize}
\item
$X$ is locally polynomial-time.
\item
There is a locally polynomial-time chain complex $\EC_*$,
also parameterized by $\II$, such that, for each fixed $\thedim$,
the distinguished basis $\Bas(I)_\thedim$ of $\EC(I)_\thedim$
can be computed in time polynomial in $\size(I)$, and in particular,
the rank $r(I)_\thedim$
is bounded by such a polynomial.
\item For every $I\in\II$, there is a reduction
$\rho_I$ from $C_*(X(I))$ to $\EC(I)_*$,
where the maps $(f_I)_\thedim,(g_I)_\thedim,(h_I)_\thedim$
of $\rho_I$ are all computable in time bounded by a polynomial
in $\size(I)$ plus the size of the input $\thedim$-chain;
the polynomial may depend on~$\thedim$.
\end{itemize}



\section{Polynomial-time homology
from a discrete vector field}\label{s:vecfield}

Discrete Morse theory, developed by Forman \cite{Forman-discrmorse}
(also see \cite{Forman-discrmorse-guide}),
belongs among fundamental tools in combinatorial topology.
For us, the important point is that a suitable
\emph{discrete vector field} on a simplicial
 set\footnote{In \cite{RoSe-vecfields}, vector fields are considered in
somewhat greater generality, on \emph{algebraic cell complexes}.
Here it is sufficient to stay in the perhaps more intuitive setting of
vector fields on simplicial sets.} $X$ can be used
to equip $X$ with effective homology;
this is an implication of one of Forman's results, as was observed by
Romero and Sergeraert \cite{RoSe-vecfields} (they also generalized
Forman's construction by dropping a certain finiteness condition).
Here we review
the definitions, more or less repeating in a general setting
the definitions given for $K(\Z,1)$ in the introduction.
Then we formulate a sufficient condition on
the vector field so that the construction provides
polynomial-time homology for~$X$.

\heading{Discrete vector fields. }
Let $X$ be a simplicial set. For a simplex $\tau\in X$,
it may happen that two face operators give the same simplex,
i.e., $\partial_i\tau=\partial_j\tau$, $i\ne j$ (geometrically,
this means that the two faces of the simplex $\tau$ are
``glued together''). We say that $\sigma$ is a \emph{regular face}
of $\tau$ if $\sigma=\partial_i\tau$ for \emph{exactly one} index~$i $.

 A \emph{discrete vector field} $V$
on a simplicial set $X$ is a set of ordered pairs (directed edges)
of the form $(\sigma,\tau)$, where $\sigma,\tau\in X^{\ndg}$,
$\sigma$ is a regular face of $\tau$,
and for every two distinct pairs
$(\sigma,\tau),(\sigma',\tau')\in V$,
all of $\sigma,\tau,\sigma',\tau'$ are distinct.

Given a discrete vector field $V$, the nondegenerate simplices of $X$ are
classified into three subsets $\SS,\TT$, and $\CC$ as follows:
\begin{itemize}
\item $\SS$ are the \emph{source simplices};
these are simplices $\sigma$ such that
$(\sigma,\tau)\in V$ for some $\tau$.
\item $\TT$ are the \emph{target simplices};
these are simplices $\tau$ such that
$(\sigma,\tau)\in V$ for some $\sigma$.
\item$\CC$ are the  \emph{critical simplices};
these are the remaining simplices,
not occurring in any edge of~$V$.
\end{itemize}

Often it is useful to regard $V$ as a bijective mapping $V\:\SS\to\TT$,
as we did in the introduction.
Thus, for $(\sigma,\tau)\in V$, we sometimes write $\tau=V(\sigma)$
and $\sigma=V^{-1}(\tau)$.

In a drawing of a simplicial set, the pairs $(\sigma,\tau)$
of a vector field can be indicated by arrows pointing from $\sigma$
into $\tau$, as in Fig.~\ref{f:pplane}.

\labfig{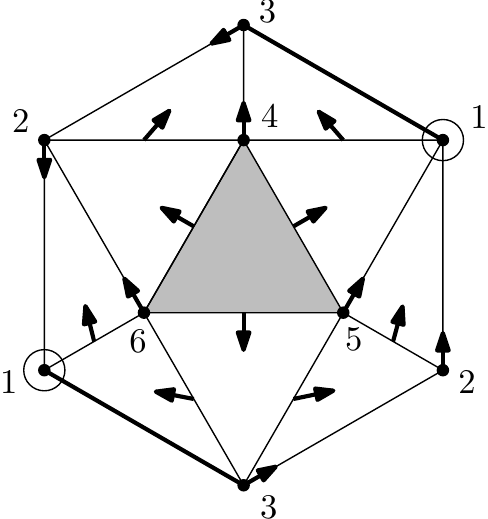}{A triangulation of the real projective plane
with a discrete
vector field (after Forman \cite{Forman-discrmorse-guide}, Fig.~4.1).
Pairs of vertices with the same label should be identified;
thus, there are only one critical edge and one critical vertex.}

\labfig{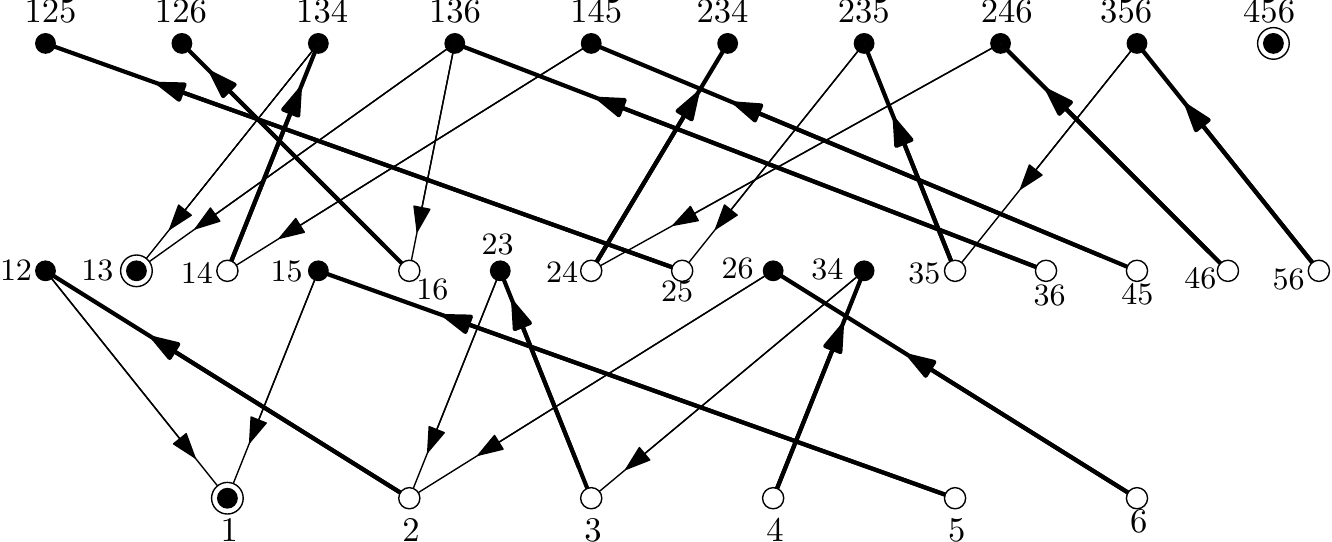}{The $V\!\partial$-graph corresponding
to Fig.~\ref{f:pplane}.}

\heading{Admissible vector fields and the $V\!\partial$-graph.}
The vector fields useful in discrete Morse theory, as well
as in our context, have an extra property.
For defining it,  we first introduce an auxiliary directed
graph, as drawn in Fig.~\ref{f:pplane1}, which we
call the $V\!\partial$-graph.

The vertex set of the $V\!\partial$-graph is $X^{\ndg}$.
In the drawing, the empty circles correspond to source simplices,
the full circles to target simplices, and the critical simplices
are marked by double circles.

The edges of the $V\!\partial$-graph
are of two kinds: first, those belonging to $V$ (drawn bold and pointing
upwards), and second, all edges of the form $(\tau,\sigma)$, where
$\tau$ is a \emph{target} simplex, $\sigma$ is a face of $\tau$
and a \emph{source} or \emph{critical} simplex, and $(\sigma,\tau)\not\in V$
(these edges point downwards).\footnote{In a simplicial
set, it may happen that $\sigma$ is a ``multiple'' face of
$\tau$. i.e., $\sigma=\partial_i\tau$ holds for several indices~$i$.
In such case, we connect $\tau$ to $\sigma$ with multiple edges in
the $V\!\partial$-graph, one edge for each such index~$i$.}
 These edges correspond to the ``allowed
moves'' defined in the introduction.

We call the vector field $V$ \emph{admissible} if the $V\!\partial$-graph
contains no directed cycle and no infinite directed path. The field
in Fig.~\ref{f:pplane} is admissible, for example.

One of Forman's results says that an admissible
 vector field $V$ can be used to
``simplify'' the underlying simplicial set $X$: by a sequence of
suitable collapsing operations, which is defined based on $V$,
one obtains a cell complex (no longer necessarily a simplicial set),
which is homotopy equivalent to $X$ but typically much smaller---its
cells correspond only to the critical simplices.

We will not use this result directly (and thus we don't formulate
it precisely). Rather, we build on a related result (obtained implicitly by
Forman with an additional finiteness assumption, and
explicitly and in general in \cite{RoSe-vecfields}),
 asserting that an admissible vector field provides
a reduction of the normalized chain complex $C_*(X)$ to a
suitable chain complex $C_*^{\rm crit}$. In this chain
complex, each $C_\thedim^{\rm crit}$ is the free Abelian
group on the set of all $\thedim$-dimensional critical
simplices. The differentials in $C_*^{\rm crit}$ are defined
based on~$V$, and they are locally effective assuming that
$X$ and $V$ are locally effective in a natural sense.

\heading{Polynomially bounded vector fields. }
We need a polynomial-time version of this result.
Let $V$ be an  admissible vector field $V$ on
a locally polynomial-time simplicial set; we assume that
both $X$ and $V$ are parameterized by a set $\II$,
as in the definition of a locally polynomial-time simplicial
set.\footnote{Of course, for the main result of this paper, polynomial-time
homology for $K(\Z,1)$, parameterization is not needed, but
we need it if we want to have a general tool for obtaining
polynomial-time homology from a vector field.}
For $\sigma\in X^\ndg$, let $\reach_V(\sigma)$
(or just $\reach(\sigma)$ if $V$ is understood) denote
the set of all simplices reachable from $\sigma$ by a directed
path in the $V\!\partial$-graph.

Let us say that $V$ is \emph{polynomially bounded} if
the following hold:
\begin{enumerate}
\item[(PBV1)] An algorithm is available that, given $I\in\II$ and a simplex
$\sigma\in X(I)^\ndg_\thedim$, classifies $\sigma$ as source, target,
or critical. In the source case, it also returns the
simplex $V(\sigma)$. The running time
is polynomial in $\size(I)+\size(\sigma)$ for every fixed~$\thedim$.
\item[(PBV2)] For every fixed $\thedim $ and every
$\sigma\in X(I)^\ndg_\thedim$, the sum of encoding sizes of all simplices
in $\reach_V(\sigma)$ is bounded by a polynomial in
$\size(I)+\size(\sigma)$.
\end{enumerate}

\begin{theorem}\label{t:PBVtoPTH} If $X$ is a locally polynomial-time
simplicial set and $V$ is a polynomially bounded vector field on $X$
such that, for every $\thedim $, the sum of the encoding sizes
of all  $\thedim $-dimensional critical simplices is polynomially
bounded (in $\size(I)$), then $X$ can be turned into a simplicial set with
polynomial-time homology.
\end{theorem}


\begin{proof}
The proof essentially follows by inspecting the work of
Forman \cite{Forman-discrmorse} (mainly Sections~7 and~8) and
making simple observations about the computation of the
relevant maps. For the reader's convenience, we provide a
self-contained presentation; this seems simpler and not much
longer than referring to the appropriate claims in Forman's
paper, introducing his notation, etc. Our presentation is,
similar to that of Forman, mainly in a combinatorial
language. We refer to \cite{RoSe-vecfields} for two other,
more algebraic variants of essentially the same proof.

Throughout the proof, we keep the parameterization of $X$ and $V$
by $\II$ implicit.

To provide the desired reduction from $C_*:=C_*(X)$,
we need to define the target chain complex $C_*^{\rm crit}$ and provide
the three maps $f,g,h$ as in the definition of a reduction.
We begin with introducing several auxiliary maps and checking some
of their properties.

The vector field $V$ induces a sequence
$V_{\#}=(V_{\#k})_{k=-\infty}^\infty$ of homomorphisms
 $V_{\#k}\:C_k\to C_{k+1}$, as follows: for a source $k$-simplex $\sigma$,
we have $V_{\#k}(1\cdot\sigma):=(-1)^{i+1}\cdot V(\sigma)$, where
$i$ is the unique index with $\sigma=\partial_i V(\sigma)$,
and for $\sigma$ target or critical, we have $V_{\#k}(1\cdot\sigma):=0$.

Next, we introduce a chain map $\Phi\:C_*\to C_*$ by
$$
\Phi:=1 +V_{\#}\compos \diff +\diff\compos V_{\#},
$$
where $1$ stands for the identity chain map and $\diff$ is the differential
of $C_*$.  It is easy to check that $\Phi$ is a chain map: indeed,
$d\compos\Phi=\diff+\diff\compos V_{\#}\compos \diff+\diff\compos \diff\compos V_{\#}=
\diff+\diff\compos V_{\#}\compos \diff=\Phi\compos \diff$ (using $\diff\compos\diff=0$).

For the proof, it is important to understand how $\Phi$ works.
We will thus discuss how the image $\Phi(1\cdot\sigma)$ is formed,
depending on the type of a $k$-simplex~$\sigma$.
\begin{enumerate}
\item The simplest case is $\sigma$ a \emph{target} simplex;
see Fig.~\ref{f:mapPhi} left. Then $V_{\#}(1\cdot\sigma)=0$, and thus
$\Phi(1\cdot\sigma)=1\cdot\sigma+
\sum_{i=0}^k V_{\#k-1}((-1)^i\cdot\partial_i\sigma)$.
So we consider all faces $\sigma'$ of $\sigma$, with the appropriate signs,
and apply $V_{\#}$ to them. Only the $\sigma'$ that are sources may
contribute to the image (and then $(\sigma,\sigma')$
are edges of the $V\!\partial$-graph), and $\Phi(1\cdot\sigma)$ is supported
only on target simplices.

Moreover, we observe that, crucially, the coefficient of $\sigma$
in $\Phi(1\cdot\sigma)$ is $0$; indeed, if $j$ is the unique
index with $V^{-1}(\sigma)=\partial_j \sigma$, then we have
$V_{\#}((-1)^j\cdot\partial_j\sigma)=(-1)^{j+1}(-1)^j\cdot\sigma=-1\cdot\sigma$,
which cancels out with the $1\cdot\sigma$ coming from the $1$
in the definition of $\Phi$. (Here we rely on the condition that
$V^{-1}(\sigma)$ is a \emph{regular} face of $\sigma$ from the definition
of discrete vector field, since we need the coefficient of
$V^{-1}(\sigma)$ in $d(1\cdot\sigma)$ to be invertible, i.e., equal to
$\pm 1$.)

Summarizing, $\Phi(1\cdot\sigma)$ consists of the target simplices
reachable from $\sigma$ in exactly two steps in the $V\!\partial$-graph,
with appropriate signs.
\item For $\sigma$ a \emph{critical} simplex we find, by a similar
reasoning, that $\Phi(1\cdot\sigma)$ consists of $\sigma$ with coefficient~$1$,
plus all the (target) simplices reachable from $\sigma$ in exactly
two steps in the $V\!\partial$-graph, again with appropriate signs.
\item Finally, for $\sigma$ a source, both the $\diff\compos V_{\#}$ and
$V_{\#}\compos\diff$ terms may make a nonzero contribution to $\Phi(1\cdot\sigma)$.
For $\diff\compos V_{\#}$ (going first up, then down), we get, with appropriate signs,
 all the \emph{source} simplices reachable from $\sigma$ in exactly
two steps in the $V\!\partial$-graph, with $\sigma$ itself cancelled out,
plus some additional target and critical simplices (here we do not follow
the edges of the $V\!\partial$-graph---that's why the arrows are dotted
in the picture).
For $V_{\#}\compos \diff$ (first down, then up), we get only target simplices.
\end{enumerate}

\labfig{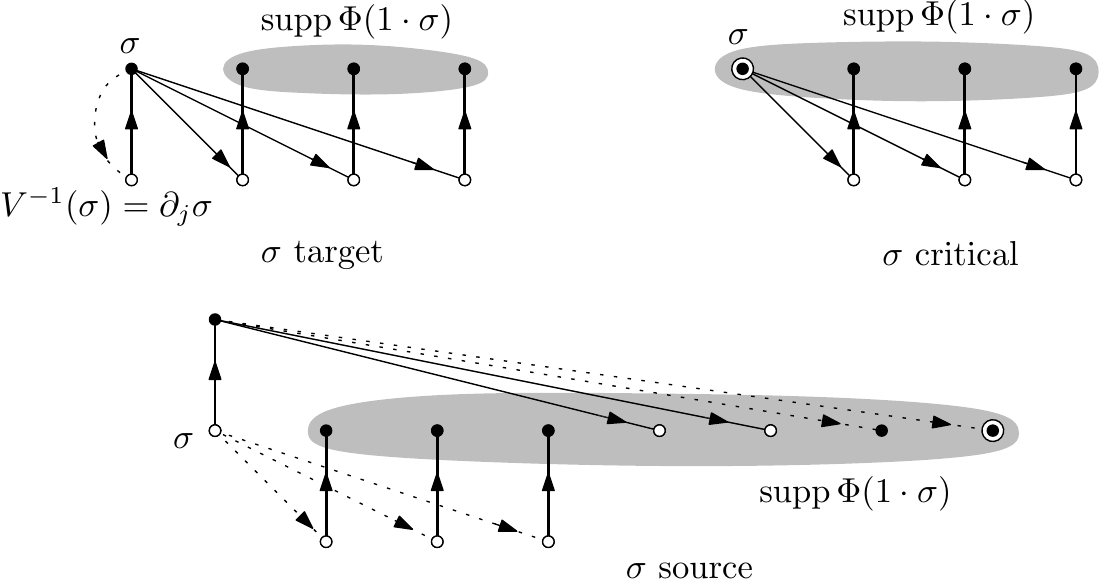}{Forming the image $\Phi(1\cdot\sigma)$.}

Next, we define $\Phi^\infty=\lim_{N\to\infty}\Phi^N$ as the
\emph{stabilization} of $\Phi$; that is, given a $k$-chain $c$,
we compute $\Phi(c)$, $\Phi(\Phi(c))$, etc., until we reach a chain
$\tilde c$ with $\Phi(\tilde c)=\tilde c$, and we set
$\Phi^\infty(c):=\tilde c$.

To check that the iterations of $\Phi$ indeed stabilize after finitely many steps, it suffices to consider the case $c=1\cdot\sigma$, and then the stabilization
follows easily from the above discussion of the action of $\Phi$
(and from the admissibility of the vector field $V$).
Moreover, we can see that the chains in $\im\Phi^\infty$ are supported only on
critical and target simplices.

We also need to check that $\Phi^\infty$ is computable in polynomial
time. In order to compute $\Phi^\infty(1\cdot\sigma)$ (which is sufficient),
we just compute the iterations $\Phi^N(1\cdot\sigma)$, $N=1,2,\ldots$,
 until they stabilize. We note that each simplex in the support of
some $\Phi^N(1\cdot\sigma)$ can be reached from $\sigma$ by following
a directed path in the $V\!\partial$-graph, then possibly
going to a face of the current simplex (a step corresponding
to a dotted arrow in Fig.~\ref{f:mapPhi}), and then  again
following a directed path in the $V\!\partial$-graph.
Hence, by the polynomial boundedness of the vector field $V$,
the stabilization occurs for $N$ at most polynomially large,
and the sum of the encoding sizes of all simplices
in the supports of all chains encountered along the way
is also polynomially bounded (essentially by the square of the bound
in condition (PBV2)).

Each \emph{coefficient} in the chain
$\Phi^{N+1}(1\cdot\sigma)$ is the sum of $O(\thedim)$ coefficients
in $\Phi^{N}(1\cdot\sigma)$. So each coefficient in $\Phi^{N}(1\cdot\sigma)$
is bounded by $\exp(O(N))$, and hence its size (number of bits)
is at most $O(N)$. Therefore, $\Phi^\infty$ is indeed
polynomial-time computable.

Now we define an auxiliary chain complex $C_*^\Phi$; we set
$C_\thedim^\Phi:=\im\Phi_\thedim^\infty\subseteq C_k$. Equivalently,
as is easily seen, $C_\thedim^\Phi=\{c\in C_\thedim:\Phi(c)=c\}$.
The differential of $C_\thedim^\Phi$ is the restriction of
the differential of $C_*$ (this works since $\Phi$ is a chain map).
Let $i\:C_\thedim^\Phi\to C_*$ be the inclusion (which is a chain map).

Next, we come to the definition of $C_*^{\rm crit}$; as was announced
above, the chain group $C_\thedim^{\rm crit}$ is the free Abelian group
($\Z$-module) with the set of the $k$-dimensional critical simplices
in $X$ as a basis. It remains to define the differential.

First we let $j_\thedim\:C_\thedim^\Phi\to C^{\rm crit}_\thedim$ be the
homomorphism that restricts a chain $c\in C_\thedim^\Phi$ to the critical
simplices (i.e., for $c=\sum_{\sigma\in X_\thedim}\alpha_\sigma\cdot
\sigma$, we set $j_\thedim(c):=\sum_{\sigma\in X_\thedim\cap\CC}\alpha_\sigma\cdot
\sigma$). We observe that $\Phi^\infty_\thedim$, viewed as a homomorphism
$C^{\rm crit}_\thedim\to C_\thedim^\Phi$, is an inverse to $j_\thedim$.
Indeed, from the description of $\Phi$ given above, it is easy to see
that for $\sigma$ critical, $\Phi^\infty(1\cdot\sigma)=1\cdot\sigma+c'$
for some $c'$ supported on target simplices, and from this the claim
follows.

Hence each $C^{\rm crit}_\thedim$ is isomorphic to $C^\Phi_\thedim$,
and the differential $\diff^{\rm crit}$ of $C^{\rm crit}_*$ can
be defined so as to make $j$ and $\Phi^\infty$ mutually inverse
chain isomorphisms;
explicitly, $\diff^{\rm crit}:=j\compos \diff\compos\Phi^\infty$.
This finishes the definition of the target chain complex
for the desired reduction; it is clear that the matrices of the
differential $\diff^{\rm crit}$ are polynomially computable,
provided that the total encoding size  of the critical simplices
is polynomial in each dimension.

It remains to define the maps $f,g,h$ in the reduction.
The following diagram summarizes the relevant chain complexes
and maps defined so far, plus $f,g,h$:
\immfig{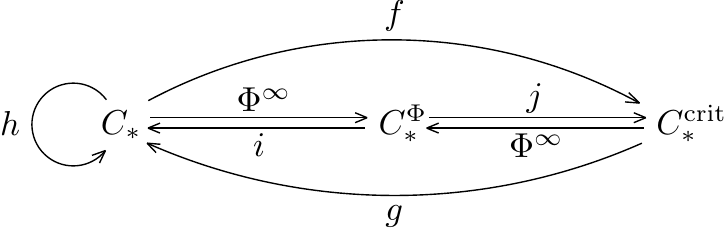}
As the diagram suggests, we put $f:=j\compos\Phi^\infty$ and
$g:=i\compos\Phi^\infty$.  Then, since $j$ and $\Phi^\infty$ are mutually
inverse and $\Phi^\infty\compos i=1$, we have $f\compos g=1$, as required
by the definition of a reduction, and $g\compos f=i\compos \Phi^\infty$.

The chain homotopy $h$ of $i\compos\Phi^\infty$ with
the identity (Forman uses the letter $L$ for this map) is now
defined as the stabilization of the maps
$$
-V_{\#}\compos(1+\Phi+\Phi^2+\cdots+\Phi^N), \ \ N=1,2,\ldots
$$
To see that these iterations indeed stabilize on each chain
$1\cdot\sigma$, we recall that for sufficiently large $N$,
$\Phi^N(1\cdot\sigma)$ is supported only on critical and target simplices,
and $V_{\#}$ sends such chains to~$0$. By essentially the same
argument as that for the computability of $\Phi^\infty$, we
also get that each $h_\thedim$ is computable in polynomial time.

We need to verify that $h$ is the required chain homotopy, i.e.,
$d\compos h+h\compos d=1-i\compos\Phi^\infty$. This is a simple formal
calculation (showing where the formula for $h$ comes from),
which we leave to the reader
(also see \cite[proof of Th.~7.3]{Forman-discrmorse}).

As the last step, we want to check the conditions $f\compos h=0$, $h\compos g=0$,
and $h\compos h=0$. To this end, we note that the chains in $\im h$
are supported only on target simplices.
Moreover, if $c$ is a chain supported only
on target and critical simplices, then $\Phi(c)$ has the same property,
and hence $h(c)=0$. These two properties immediately give $h\compos h=0$.
Similarly, $\im g=\im\Phi^\infty$ is supported only on target and critical
simplices, and hence $h\compos g=0$. Finally, we have seen that
$\Phi^\infty$ maps target simplices to $0$, and so does $f=j\compos \Phi^\infty$,
which gives $f\compos h=0$ and concludes the proof of Theorem~\ref{t:PBVtoPTH}.
\end{proof}

\section{A polynomially bounded vector field for $K(\Z,1)$}\label{s:KZ1}

Here we finally get to the combinatorial core of the paper;
we will provide a polynomially bounded vector field for $K(\Z,1)$.

\heading{A simple composition of vector fields. }
For the sake of presentation, it will be easier to split the
vector field into two parts. Roughly speaking, the first part
will get rid of all negative components in the considered
sequences $[a_1|\cdots|a_\thedim]$, and the second part will do the rest.

Here is the way of ``splitting into two parts''
in a general setting. Let $X$ be a simplicial
set, let $V_1$ be a vector field on $X$, with the
set $\CC_1$ of critical simplices, and suppose that $\CC_1$
is closed under the face operators (each face of a critical simplex is
again critical, or degenerate).
Let $Y$ be the simplicial subset of $X$ induced by $\CC_1$
(i.e., its nondegenerate simplices are the critical simplices of $V_1$),
and let $V_2$ be a vector field on $Y$.

Then we can define a ``composition'' $V$ of $V_1$ and $V_2$ in the
obvious way; formally, if we regard a vector field a set of
ordered pairs, we simply set $V:=V_1\cup V_2$. Clearly, $V$
is a vector field, and it is easily seen that $V_1,V_2$
admissible imply $V$ admissible, and similarly for
polynomial boundedness.

In the case of $X=K(\Z,1)$, the role of $Y$ will be played by
the simplicial set whose simplices are the integer sequences
with all terms \emph{nonnegative}. With some abuse of the usual
notation, we will denote this simplicial set by $K(\N,1)$.

The first vector field will be denote by $\Vbs$ and
called the \emph{bubblesort field}, since directed paths in its
$V\!\partial$-graph resemble the computation of a sorting
algorithm called Bubblesort. Its critical simplices are
integer sequences with all entries positive.

The second vector field is defined on $K(\N,1)$, and it
has only two critical simplices $[\,]$ and $[1]$,
the same as the Eilenberg--MacLane field~$V_{\rm EML}$.
We call it the \emph{bit-chipping field} and denote
it by~$\Vbc$.
\medskip

Let us remark that one can consider composition of vector
fields in a more general and more flexible setting,
as is done in \cite{RoSe-vecfields},
but for our purposes, the simple notion above suffices.

\subsection{The bubblesort field}

\heading{Translating positive sequences to sorted sequences. }
To define the vector field $\Vbs$, it is convenient to consider
a different representation of the simplices of $K(\Z,1)$.
Namely, we represent a
$\thedim $-dimensional simplex $\sigma=[a_1|\cdots|a_\thedim]$ by
a $(\thedim+1)$-tuple $(b_0,b_1,\ldots,b_\thedim)$, where
$b_0\in\Z$ can be chosen arbitrarily  and
$b_i:=b_{i-1}+a_i$, $i =1,2,\ldots,\thedim$. Thus,
each $\sigma$ is represented as an equivalence class
of $(\thedim+1)$-tuples of integers, where two $(\thedim+1)$-tuples
are equivalent if their difference
is of the form $(a,a,\ldots,a)$ (all components
equal).
We denote the equivalence class of $(b_0,\ldots,b_\thedim)$
by $[b_0,\ldots,b_\thedim]$.

This correspondence between simplices of the form
$[a_1|\cdots|a_\thedim]$ and equivalence classes of $(\thedim+1)$-tuples
is obviously bijective. Nondegenerate simplices $[a_1|\cdots|a_\thedim]$,
i.e., those with no zero component, translate to
 $[b_0,\ldots,b_\thedim]$ with $b_{i-1}\ne b_{i}$, $i =1,2,\ldots,\thedim$.

A (nondegenerate) simplex from $K(\N,1)$ corresponds to
$[b_0,\ldots,b_\thedim]$ with strictly increasing components,
i.e., $b_0<b_1<\cdots<b_\thedim$. The face operators become extremely
simple in this notation: $\partial_i$ corresponds to
deleting the $i$th component.

\heading{The field. } As was already announced, the critical
simplices of $\Vbs$ are the $[b_0,\ldots,b_\thedim]$ with
$b_0<\cdots<b_\thedim$. If $\sigma=[b_0,\ldots,b_\thedim]$ is not
critical, we look at the smallest $\ell$ such that $b_\ell>b_{\ell+1}$;
let us call it the \emph{leading index} of $\sigma$.
Let us write $v=b_\ell$ and $u=b_{\ell+1}$. We consider the maximal
contiguous segment in the sequence $b_0,b_1,\ldots$ starting
at the $\ell$th position and containing only $v$'s and $u$'s;
formally, we take the largest $m\ge \ell+1$ such
that $b_i\in\{u,v\}$ for all $i=\ell,\ell+1,\ldots,m$, and
either $b_{m+1}\not\in\{u,v\}$ or $m=\thedim$.
We call $b_\ell,b_{\ell+1},\ldots,b_m$ the \emph{leading alternating
segment} of $\sigma$ (indeed, there can be no two consecutive
$u$'s or $v$'s, since this would mean that $\sigma$ is degenerate),
and we denote it by $\LAS(\sigma)$.

Then we let $\sigma$ be a source if $\LAS(\sigma)$ ends with $u$,
 and otherwise, $\sigma$ is a target.
For a source $\sigma$, still with $u,v,m$ as above, we set
\begin{equation}\label{e:vbs}
\tau=\Vbs(\sigma):=[b_0,\ldots,b_m,v,b_{m+1},\ldots,b_\thedim],
\end{equation}
i.e., $\Vbs$ inserts another $v$ just after $\LAS(\sigma)$.

With $\tau=\Vbs(\sigma)$ as in the just given definition,
we have $\sigma=\partial_{m+1}\tau$, and $m+1$ is easily seen
to be the  only index $i$ with $\sigma=\partial_i\tau$
(thus, $\sigma$ is a regular face of $\tau$). Moreover, $\sigma$
can be uniquely reconstructed from $\tau$ (delete the
last element of  $\LAS(\tau)$), and so
 $\Vbs$ is indeed a discrete vector field.


Next, we observe that once we show that $\Vbs$ is admissible,
it becomes obvious that it is also polynomially bounded.
This is because the boundary operators only delete components
and the vector field duplicates them, and so any simplex reachable
from a given $k$-dimensional $\sigma$ is made of the components
of $\sigma$. Hence at most $(k+1)^{k+1}$ distinct source simplices
are reachable from $\sigma$, which is a constant for $k$ fixed.

It remains to prove admissibility, which is tricker than it might seem.
Let us consider a source simplex $\sigma=[b_0,\ldots,b_{\ell-1},v,\ldots,u,
b_{m+1},\ldots,b_k]$, $b_0<b_1<\cdots < b_{\ell-1}<v>u$,
where the part between the $v$ and $u$ is the $\LAS$.
We set $\tau=\Vbs(\sigma)$,
and ask for which $i$'s
the simplex $\sigma'=\partial_i\tau$ can again be a source simplex
(in this case we say that $\sigma'$ arises from $\sigma$
by a \emph{double move}).

If $\LAS(\sigma')=\LAS(\tau)$, then $\sigma'$ is a target simplex,
and  so $\partial_i$ must change $\LAS(\tau)$.
It cannot delete elements from the middle of $\LAS(\tau)$, since the result
would be degenerate, and it cannot delete the final $v$, since this
was inserted by~$\Vbs$.

Thus, one possibility is $i=\ell$, in which case $\sigma'$ is
obtained from $\sigma$ by appending $v$ to the end of the $\LAS$
and deleting the initial $v$ of the $\LAS$. Let us call this
a \emph{switching double move}. This is the ``intended'' type of double moves
that do the bubble-sorting, provided that the $\LAS$ has length~2;
for example, $\sigma=[3,1,2]$ is transformed to $\sigma'=[1,3,2]$.
A switching double move may also occur for $\LAS(\sigma)$ of length 4
or more, if the deletion of the initial $v$ creates a new $\LAS$; i.e.,
if $b_{\ell-1}>u$. An example is $\sigma=[2,3,1,3,1]$,
$\sigma'=[2,1,3,1,3]$.

However, there is a second, less obvious possibility for a double move:
if the sequence $b_{m+1},b_{m+2},\ldots]$ following $\LAS(\tau)$
has the form $x,u,v,u,v,\ldots,u,y,\ldots]$, $x,y\not\in\{u,v\}$,
or the form $x,u,v,\ldots,u]$,
then  we can also have $i=m+2$. In this case, $\partial_{m+2}$
deletes the component following the $\LAS$, and produces a longer $\LAS$.
We call this an \emph{appending double move}.
For example, for $\sigma=[2,3,1,4,1,3,1]$, the switching double move
yields $\sigma'=         [2,1,3,4,1,3,1]$ and the appending one
yields $\sigma'=         [2,3,1,3,1,3,1]$.

If we follow a sequence of directed edges in the $V\partial$-graph
starting at some source
simplex $\tilde\sigma$, and if all source simplices encountered along the way
have $\LAS$ of length~2,
then the path has a bounded length, since all the double moves are
switching in this case, and each of them decreases the number of inversions
(i.e., pairs $(i,j)$ with $i<j$ and $b_i>b_j$) in the current source
simplex.

The following lemma shows that if $\LAS(\tilde\sigma)$ has length greater
than~2, then  every sequence of double moves starting at $\tilde\sigma$
finishes after a finite number of steps, and this already
implies the admissibility of~$\Vbs$. All the difficulty of the lemma
is in getting the statement right; the proof is routine.

\begin{lemma}\label{l:bs-structure}
Let $\tilde\sigma=[b_0,\ldots,b_{\ell-1},b_\ell=v,u,v,\ldots,u,\ldots]$ ,
$b_0<\cdots<b_\ell>u$,
be a source simplex with $\LAS(\tilde\sigma)$ of length
greater than $2$. Then every source $\sigma$ obtainable from $\tilde\sigma$
by a sequence of double moves has the following structure:
$[\beta_0,\beta_1,\ldots,\beta_\ell,\gamma]$,
where each $\beta_i$ is a block of length $k_i\ge 1$ starting
with $b_i$ and possibly continuing with $u,b_i,u,b_i,\ldots$
(alternations of $b_i$ and $u$, $u<b_i$),
and $\gamma$ is a possibly empty block that does not start with~$u$.
The sequence $(k_0,k_1,\ldots,k_\ell)$
has the form \[(\underbrace{1,1,\ldots,1}_j,k_j,k_{j+1},\ldots,k_\ell),\]
where $k_j\ge 2$ is even, while all of the other $k_i$ are odd, and
there is at least one $k_i\ge 3$.

In each double move of a sequence starting at $\tilde\sigma$,
either $j$ decreases, or it stays the same and $k_j$ increases.
Thus, each such sequence is finite.
\end{lemma}

\begin{proof} The initial $\tilde\sigma$ clearly has the claimed form.
Let us assume that $\sigma$ is of this form, and let a source $\sigma'$ be obtained
from it by a double move.

We have $\LAS(\sigma)=\beta_j$, of even length $k_j\ge 2$. If the double move
is switching, then \[
\sigma'=[b_0,b_1,\ldots,b_{j-1},u,\underbrace{b_{j},\ldots,u,b_j}_{k_j-1},\beta_{j+1},\ldots,\gamma].
\]
If we had $j=0$ or $b_{j-1}<u$, then $\LAS(\sigma')$ would be either
the block $b_j,\ldots,u,b_j$ of odd length $k_j-1$ (for $k_j\ge 4$),
or, for $k_j=2$, another $\beta_i$, $i>j$, of odd length $k_i\ge 3$
 (guaranteed to
exist by the inductive assumption). In both cases $\sigma'$ would
be target, and so $b_{j-1}>u$. Then $\sigma'$ has the claimed
structure $[\beta'_0,\ldots,\beta'_\ell,\gamma]$, with $j'=j-1$, $\beta'_i=
\beta_i$ for all $i\not\in\{j-1,j\}$, $\beta'_{j-1}=b_{j-1},u$
of length $k'_{j-1}=2$, and $\beta'_j$ of odd length $k'_{j}=k_j-1$.
So $j$ has decreased.

For an appending double move, we distinguish two cases. For $j<\ell$,
there is at least one more block $\beta_{j+1}$ following $\beta_j$
in $\sigma$, with $k_{j+1}\ge 3$ (since $\beta_{j+1}$ must have an $u$
to append to $\beta_j$), and we have
\[
\sigma'=[b_0,b_1,\ldots,b_{j-1},\underbrace{b_{j},\ldots,u,b_j,u}_{k_j+2},
\underbrace{b_{j+1},u,\ldots,b_{j+1}}_{k_{j+1}-2},\beta_{j+2},\ldots,\gamma].
\]
This is the claimed structure with $j'=j$, $k'_j=k_j+2$, and $k'_{j+1}=
k_{j+1}-2$.

Finally, if $j=\ell$, then $\gamma$ has to start with $x,u,\ldots$,
and here we get $j'=j=\ell$ and $k'_\ell\ge k_\ell+2$ (depending
on the number of $u,v$ alternations in $\gamma$ following~$x$).
\end{proof}

\heading{A lower bound.}  Although the bubble-sorting process itself is only quadratic, it turns out that $|\reach_{\Vbs}(\tilde\sigma)|$ 
for a suitable source simplex $\tilde\sigma$ may indeed be exponential in $k$, and thus the bound $(k+1)^{k+1}$ claimed above is not so far off the mark.
Mainly to illustrate the behavior of the vector field $\Vbs$, 
we indicate the lower bound via a concrete example without proof. 
Namely, from 
\[
\tilde\sigma=[2,3,4,5,6,7,1,7,1,7,1,7,1,7,1,7,1,7,1]
\]
 we can reach source simplices such as 
$[2,1,2,1,3,1,3,4,5,1,5,1,5,6,7,1,7,1,7]$. Such simplices
have $6$ blocks (denoted by $\beta_0,\ldots,\beta_5$ in the proof above), 
and we can choose the block lengths at will, with the obvious restrictions
(the total length is fixed, and the block lengths are all odd except
for the first one). In an analogous construction with $6$ replaced by
an arbitrary integer $b$ we take $k=3b$ and  obtain a lower bound
exponential in~$k$.

\subsection{The bit-chipping field}

Here we return to the ``bar'' notation $[a_1|a_2|\cdots|a_\thedim]$,
and we will consider only simplices of $K(\N,1)$, which means
$a_i\ge 1$ for all~$i$.

\heading{The anatomy of a simplex. }
Let $\sigma=[a_1|a_2|\cdots|a_\thedim]$ be a nondegenerate $\thedim $-simplex of $K(\N,1)$.
We introduce the following terminology.
\begin{itemize}
\item Let $p=p(\sigma)\in\{0,1,\ldots,\thedim\}$
be the largest index such that
$a_1,\ldots,a_p$ are all powers of $2$ and
$a_1\le a_2\le\cdots\le a_p$. The sequence
$a_1|a_2|\cdots|a_p$ is called the \emph{nondecreasing dyadic part} of $\sigma$.
If $1\le p<\thedim$ and $a_{p}>a_{p+1}$, then $p$ is called
the \emph{peak} of $\sigma$; otherwise, $\sigma$ has no peak.
\item Let $q=q(\sigma)\in\{0,1,\ldots,\thedim\}$ be the largest index such that
$a_1,\ldots,a_q$ are all powers of $2$ (thus, $q\ge p$). The sequence
$a_1|a_2|\cdots|a_q$ is called the \emph{dyadic part} of $\sigma$.
If $q=\thedim$, then $\sigma$ is called \emph{fully dyadic}.
If, on the other hand, $q<\thedim$, then $q+1$ is the
\emph{breakpoint} of $\sigma$ and $a_{q+1}$ is the
\emph{breakpoint value} of $\sigma$ (which is not a power of~$2$).
The sequence $a_{q+2}|a_{q+3}|\cdots|a_\thedim$ is the \emph{right part} of~$\sigma$.
\end{itemize}

Here are two concrete examples:
\immfig{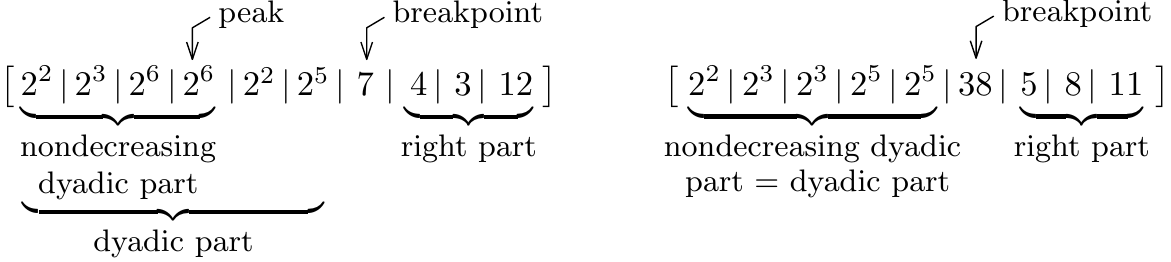}

\heading{The vector field. } We define a vector field $\Vbc$ on $K(\N,1)$.
There are two types of source simplices.
\begin{enumerate}
\item[(a)] The first type of source simplices
are the simplices that are not fully dyadic and have no peak.
Thus, all of the dyadic part is nondecreasing (i.e., $p=q$;
we also admit $p=q=0$)
and the breakpoint value is larger than the last element of the dyadic part.
Explicitly, they are of the form
$$
\sigma=\left[2^{i_1}\,|\,2^{i_2}\,|\,\cdots\,|\,2^{i_q}\,|\,b\,|\,a_{q+2}\,|\,
\cdots\,|\,a_\thedim\right],
$$
$2^{i_1}\le 2^{i_2}\le\cdots\le 2^{i_q}<b$. In this case we set
\begin{equation}\label{e:tau(a)}
\Vbc(\sigma)=\tau:= \left[2^{i_1}\,|\,2^{i_2}\,|\,\cdots\,|\,2^{i_q}\,|
\lpow(b)\,|\,\ltrim(b)\,|\,a_{q+2}\,|\,\cdots\,|\,a_\thedim\right],
\end{equation}
where $\lpow(b)$ is the largest power of $2$ not exceeding $b$,
and $\ltrim(b):=b-\lpow(b)$. That is, $\tau$ is obtained by
splitting the breakpoint value $b$ into two components,
$\lpow(b)$ and $\ltrim(b)$; informally, we can think of this as
``chipping off'' the leading bit of~$b$.

We observe that each target simplex $\tau$ as defined above
has a peak, namely, $p(\tau)=q(\sigma)+1$,
and in particular, $\tau$ has a nonempty dyadic part
 (but it may happen that
the dyadic part of $\tau$ is longer than the nondecreasing dyadic part,
since $\ltrim(b)$ may be a power of two).
\item[(b)] The second type of source simplices
are the fully dyadic simplices
$
\sigma=\left[2^{i_1}\,|\,2^{i_2}\,|\,\cdots\,|\,2^{i_\thedim}\right]
$
with $2^{i_1}\le 2^{i_2}\le\cdots\le 2^{i_{\thedim-1}}<2^{i_\thedim}$
with $i_\thedim\ge 1$ (this last condition is important only for $\thedim =1$).
In this case we set
\begin{equation}\label{e:tau(b)}
\tau=\Vbc(\sigma):=\left[2^{i_1}\,|\,2^{i_2}\,|\,\cdots\,|\,2^{i_{\thedim-1}}\,|\,
2^{i_\thedim-1}\,|\,2^{i_\thedim-1}\right];
\end{equation}
 i.e., we split the last component
of $\sigma$ into two equal halves.
\end{enumerate}

\begin{lemma}\label{l:validV} This definition indeed yields a vector
field, and the only critical simplices are $[\,]$ and $[1]$.
\end{lemma}

\begin{proof} Let us consider an arbitrary simplex $\tau$.
If it is not fully dyadic and is not a source simplex, then
it has a peak, and thus it has the form
$\tau=[2^{i_1}\,|\cdots|\,2^{i_p}\,|\,c_{p+1}\,|\cdots|\,c_{\thedim+1}]$
with $2^{i_1}\le\cdots\le 2^{i_p}>c_{p+1}$. This equals $\Vbc(\sigma)$
for $\sigma=[2^{i_1}\,|\cdots|\,2^{i_{p-1}}\,|\,2^{i_p}+c_{p+1}\,|\,
c_{p+2}\,|\cdots|\,c_{\thedim+1}]$. Thus, $\tau$ is a target simplex
and there is exactly one edge $(\sigma,\tau)\in \Vbc$.
Moreover, we have $\sigma=\partial_p\tau$, while
$\partial_j\tau\ne\sigma$ for $j\ne p$, so $\sigma$ is a
regular face of~$\tau$ as needed.

Next, if $\tau$ is fully dyadic and has a peak $p$, i.e.,
$\tau=[2^{i_1}\,|\cdots|\,2^{i_{\thedim+1}}]$,
$2^{i_1}\le\cdots\le 2^{i_p}>2^{i_{p+1}}$, then
$\tau$ is again a target simplex with $\tau=\Vbc(\sigma)$ for
$\sigma=[2^{i_1}\,|\cdots|\,2^{i_{p-1}}\,|\,2^{i_p}+2^{i_{p+1}}\,|\,2^{i_{p+1}}\,|\cdots|\,2^{i_{\thedim+1}}]$ (here $2^{i_p}+2^{i_{p+1}}$ is the breakpoint
value). Again, $j=p$ is the only index with $\partial_j\tau=\sigma$.

The last remaining case is a fully dyadic $\tau$ with no peak,
which must be nondecreasing. If it is not a source simplex, then
either we have one of the cases $[\,]$, $[1]$, or $\thedim \ge2$
and the last two components of $\tau$ are equal, which means
that $\tau$ is of the form (\ref{e:tau(b)}) and $\sigma$
can again be uniquely reconstructed from it. We have
$\sigma=\partial_j\tau$ for the unique index $j=d-1$.
\end{proof}

\heading{Preparations for analyzing $\Vbc$.}
It will be convenient to work mainly with the target simplices.
Thus, given a target simplex $\tau$, we let $\treach(\tau)\subset
\reach(\tau)$ be the set of all target simplices reachable from $\tau$.

First we will classify all possible target simplices $\tau'$ reachable from
a given target simplex $\tau$ by two steps in the $V\!\partial$-graph;
in other words, the $\tau'$ of the form $\Vbc(\partial_j\tau)$
for some~$j$.
This is a straightforward, if somewhat lengthy, case analysis.
The subsequent proofs of admissibility and polynomial boundedness
will use this classification. It would be nice to avoid considering so
many cases, but one needs to be careful in the analysis:
 for several other candidate vector fields we have tried, ``most'' cases 
apparently worked fine, but those fields failed in what seemed 
like minor details.

\begin{lemma}\label{l:taucases}
 Let $\tau=[a_1|a_2|\cdots|a_\thedim]$
be a $\thedim $-dimensional target simplex.

If $\tau$ is not fully dyadic, we can write it in the form
$$
\left[2^{i_1}\,|\,2^{i_2}\,|\,\cdots\,|\,2^{i_{p}}\,|\,2^{i_{p+1}}\,|\cdots
|\,2^{i_q}\,|\,b\,|\,a_{q+2}\,|\cdots\,|\,a_\thedim\right],
$$
where $b$ is not a power of $2$,
$2^{i_1}\le\cdots\le 2^{i_p}$, $p\ge 1$, $p\le q\le \thedim-1$, and either
$2^{i_p}>2^{i_{p+1}}$ (if $p<q$) or $2^{i_p}>b$ (for $p=q$).
Let $\tau'$ be a target simplex of the form $\Vbc(\partial_j\tau)$
for some $j $, where $\sigma=\partial_j\tau$ is a $(\thedim-1)$-dimensional
source simplex. Then $\tau'$ has one of the following forms:
\begin{enumerate}
\item[\rm (\caseA)] If $p=1$ and $2^{i_{2}}\le\cdots\le 2^{i_q}<b$, then
we can have
$$\tau'=\left[2^{i_2}\,|\cdots|\,2^{i_q}\,|\,\lpow(b)\,|\,\ltrim(b)\,|\,a_{q+2}\,|\cdots|\,a_\thedim\right]$$
(we drop the first component and split~$b$).
\xample{2^2|1|2|7}{1|2|2^2|3}
\item[\rm (\caseB)] If $i_j<i_{j+1}$ for some $j $, $1\le j \le p-1$, then
we can have
$$
\tau'=\left[2^{i_1}\,|\cdots|\,2^{i_{j-1}}\,|\,2^{i_{j+1}}\,|\,
2^{i_j}\,|\,2^{i_{j+2}}\,|\cdots|\,2^{i_q}\,|\,b\,|\,a_{q+2}\,|\cdots|\,a_\thedim\right]
$$
(the entries $2^{i_j}$ and $2^{i_{j+1}}$ are swapped).
\xample{1|2^2|2|7}{2^2|1|2|7}
\item[\rm (\caseC)] If $q\ge p+2$, $i_{p}-1=i_{p+1}=i_{p+2}<i_{p+3}\le
\cdots\le i_q$, and $2^{i_q}<b$, then
we can have
$$
\tau'=\left[2^{i_1}\,|\cdots|\,2^{i_p}\,|\,2^{i_p}\,|\,2^{i_{p+3}}\,|\cdots
|\,2^{i_q}|\,\lpow(b)\,|\,\ltrim(b)\,|\,a_{q+2}\,|\cdots|\,a_\thedim\right]
$$
(two components following the peak are merged and $b$ is split).
\xample{2|1|1|2|7}{2|2|2|2^2|3}
\item[\rm (\caseD)] If $q\ge p+2$ and $i_{p+2}\ge i_p>i_{p+1}$, then
we can have
$$
\tau'=\left[2^{i_1}\,|\cdots|\,2^{i_p}\,|\,2^{i_{p+2}}\,|\,
2^{i_{p+1}}\,|\,2^{i_{p+3}}\,|\cdots|\,2^{i_q}\,|\,b\,|\,a_{q+2}\,|\cdots|\,a_\thedim\right]
$$
(the entries $2^{i_{p+1}}$ and
$2^{i_{p+2}}$ are swapped).
\xample{2|1|2^2|7}{2|2^2|1|7}
\item[\rm (\caseE)] If $q=p+1$, $b'=2^{i_{p+1}}+b$
satisfies $b'\ge2^{i_p}$, and $b'$
is not a power of $2$, then
we can have
$$
\tau'=\left[2^{i_1}\,|\cdots|\,2^{i_p}\,|\,\lpow(b')\,|\,
\ltrim(b')\,|\,a_{q+2}\,|\cdots|\,a_\thedim\right].
$$
\xample{2^3|2|7}{2^3|2^3|1}
\item[\rm (\caseEa)]
If the situation is as in (\caseE) except that
$b'=2^i$ is a power of $2$, then we can have
$$
\tau'=\Vbc\left(
\left[2^{i_1}\,|\cdots|\,2^{i_p}\,|\,2^i\,|\,a_{q+2}\,|\cdots|
\,a_\thedim\right]\right)
$$
(note that here we do not write out $\tau'$ explicitly,
since there are still several cases to distinguish depending
on the right part of $\tau$, but we will not need to discuss them explicitly).
\xample{2^3|1|7|19}{2^3|2^3|2^4|3}
\item[\rm (\caseF)] If $q=p\le \thedim-2$, $b':=b+a_{q+2}\ge 2^{i_p}$,
and $b'$ is not a power of $2$, then we can have
$$
\tau'=\left[2^{i_1}\,|\cdots|\,2^{i_p}\,|\,\lpow(b')\,|\,
\ltrim(b')\,|\,a_{q+3}\,|\cdots|\,a_\thedim\right].
$$
\xample{2^3|7|4}{2^3|2^3|3}
\item[\rm (\caseG)] If the conditions are as in {\rm (\caseF)}
except that $b'=2^i$ is a power of $2$, then we can have
$$
\tau'=\Vbc\left(\left[2^{i_1}\,|\cdots|\,2^{i_p}\,|\,2^i\,|\,a_{q+3}\,|\cdots|\,a_\thedim\right]\right)
$$
(as in (\caseEa), we need not write out $\tau'$ explicitly).
\xample{2^3|7|1|7}{2^3|2^3|2^2|3}
\item[\rm (\caseH)] If $q=p=\thedim-1$ and either $p=1$ or $i_{p-1}<i_p$, then
we can have
$$
\tau'=\left[2^{i_1}\,|\cdots|\,2^{i_{p-1}}\,|\,
2^{i_p-1}\,|\,2^{i_p-1}\right].
$$
\xample{2|2^3|7}{2|2^2|2^2}
\end{enumerate}

If $\tau=[2^{i_1}|\cdots|2^{i_\thedim}]$ is fully dyadic,
then either $p<\thedim$ ($\tau$ has a peak), or
$p=\thedim$ ($\tau$ is nondecreasing) and $i_{\thedim-1}=i_\thedim$.
In the peak case, we have
the following
possibilities for $\tau'=\Vbc(\partial_j\tau)$:
\begin{enumerate}
\item[\rm (\casedA)] If $p=1$
and $i_2\le i_3\le\cdots\le i_{\thedim-1}<i_\thedim$, we can have
$$
\tau'=\left[2^{i_2}\,|\cdots|\,2^{i_{\thedim-1}}\,|\,2^{i_\thedim-1}\,|\,2^{i_\thedim-1}
\right]
$$
(deleting the first entry of $\tau$ and splitting the last).
\item[\rm (\casedB)] For $1\le j\le p-1$ and $i_j<i_{j+1}$,
$\tau'$ can be obtained by swapping $2^{i_j}$ and $2^{i_{j+1}}$.
\item[\rm (\casedC)] If  $i_p-1=i_{p+1}=i_{p+2}<i_{p+3}\le\cdots\le i_{\thedim-1}<i_\thedim$,
we can have
$$
\tau'=\left[2^{i_1}\,|\,\cdots|\,2^{i_p}\,|\,2^{i_p}\,|\,
 2^{i_{p+3}}\,|\cdots|\,2^{i_{\thedim-1}}\,|\,2^{i_\thedim-1}\,|\,2^{i_\thedim-1}\right]
$$
(merging two equal entries and splitting the last).
\item[\rm (\casedD)] For $i_{p+2}\ge i_p>i_{p+1}$,
$\tau'$ can be obtained from $\tau$ by swapping
$2^{i_{p+1}}$ and $2^{i_{p+2}}$.
\end{enumerate}
Finally, if a fully dyadic $\tau$ has no peak, we have
the possibility {\rm (\casedB) } for $\tau'$ and the following
additional one:
\begin{enumerate}
\item[\rm (\casedH)] If $\thedim =2$ or $i_{\thedim-2}<i_{\thedim-1}$, then we can have
$$
\tau'=\left[2^{i_1}\,|\,\cdots|\,2^{i_{\thedim-2}}\,|\,
2^{i_{\thedim-1}-1}\,|\,2^{i_{\thedim-1}-1}\right]
$$
(drop the last component and split the previous one).
\end{enumerate}

\end{lemma}

\begin{proof}
As was already mentioned, the proof is
totally straightforward and could probably be left to the
reader. Yet, since getting used to the definitions
and notation probably needs some practice, we chose to
present the proof.

As in the lemma, we first consider $\tau$ not fully dyadic.
If $\sigma=\partial_j\tau$ is a source simplex, then it has no peak,
and thus the operation $\partial_j$ has to ``destroy'' the peak of
$\tau$ in some way. In particular, we have $j \le p+1$, for otherwise,
the peak of $\tau$ is also present in $\partial_j\tau$.
We just need to discuss the values of $j $ in this range.

For $j =0$, $\partial_0$ removes the first coordinate, and this
may destroy the peak only for $p=1$. For $p=1$,
$\sigma$ is a source
iff $2^{i_{p+1}}\le \cdots\le 2^{i_q}<b$ (this condition is void
for $q=1$), and if this holds, then
$\tau'$ is as in (\caseA).

If $1\le j\le p-1$, $\sigma=[2^{i_1}|\cdots|2^{i_{j-1}}|2^{i_j}+
2^{i_{j+1}}|2^{i_{j+2}}|\cdots|2^{i_q}|b|\cdots]$. In this case,
if $i_j=i_{j+1}$, then $2^{i_j}+2^{i_{j+1}}$ is a power of two,
 $\sigma$ necessarily has a peak, and thus it is not a source.
So $i_j<i_{j+1}$;  then $\sigma$ is a source and $2^{i_j}+2^{i_{j+1}}$
is the breakpoint value, and $\tau'$ is as in (\caseB).

Next, we consider $j =p$. Here the $p$th component of $\sigma$
is $2^{i_p}+2^{i_{p+1}}$ (for $q>p$) or $2^{i_p}+b$ (for $p=q$).
In both of these cases the $p$th component is not a power
of $2$ (since $p$ was the peak of $\tau$), hence $p$ is the
breakpoint of $\sigma$, and so $\Vbc(\sigma)=\tau$. Therefore,
$j =p$ does not contribute any~$\tau'$.

Finally, we need to discuss $j =p+1$. Here the sum of the two
entries of $\tau$ following the peak must greater or equal
to $2^{i_p}$ (and, in particular, $p\le \thedim-2$),
for otherwise, $p$ would be a peak in $\sigma$.
We consider three cases, depending on how many of these two
entries are powers of~$2$.

First, if $q\ge p+2$, then the peak is followed by $2^{i_{p+1}}$
and $2^{i_{p+2}}$ in $\tau$. If $2^{i_{p+1}} +2^{i_{p+2}}=2^{i_p}$,
then $i_{p+1}=i_{p+2}=i_p-1$. Then $\sigma$ begins with
$[2^{i_1}|\cdots|2^{i_p}|2^{i_p}|2^{i_{p+3}}|\cdots|2^{i_q}|b|\cdots$,
and since it has no peak, the dyadic part is nondecreasing.
Then $\tau'$ is as in (\caseC). If, on the other hand
$2^{i_{p+1}} +2^{i_{p+2}}>2^{i_p}$, then $2^{i_{p+1}} +2^{i_{p+2}}$
is not a power of $2$. Then  $\tau'$ is as in~(\caseD).

Second, we can have $q=p+1$ (still with $j =p+1$). Then the entry
of $\sigma$ following $2^{i_p}$ is $b'=2^{i_{p+1}}+b$,
 which has to be at least $2^{i_p}$.
If $b'$ is not a power of two, then
$\tau'$ is as in (\caseE), and otherwise, we get~(\caseEa).

Third, we can have $q=p$. If $q\le \thedim-2$, then the $p$th entry of $\sigma$ is
followed by $b':=b+a_{q+2}$, which has to be at
least $2^{i_p}$.  If $b'$ is not a power of two, then
$\tau'$ is as in (\caseF), and otherwise, we get~(\caseG).

There is still one remaining case for $j =p+1$, namely,
when $p=\thedim-1$; then $\partial_j$ just deletes the last
coordinate and $\sigma$ is fully dyadic. Then $\sigma$ is
a source precisely when $p=1$ or $i_{p-1}<i_p$,
and we have $\tau'$ as in~(\caseH).

\medskip

It remains to consider the case of $\tau=
[2^{i_1}|\cdots|2^{i_\thedim}]$ fully dyadic; thus,
$q=\thedim$. First we assume that $\tau$ has a peak $p\le \thedim-1$.
Then most of the analysis as above applies.

For $j=0$, we get that $\partial_0 \tau$ is  a source
iff $p=1$ and $i_2\le i_3\le\cdots\le i_{\thedim-1}<i_\thedim$,
and then we have $\tau'$ as in (\casedA).

For $1\le j\le p-1$, arguing as in the not fully dyadic case above,
for $i_j<i_{j+1}$ we get $\tau'$ by swapping $2^{i_j}$
and $2^{i_{j+1}}$ as in (\casedB). The case $j=p$ again brings no~$\tau'$.

For $j =p+1$, we have essentially the first of the three cases
of the analogous analysis for the  not fully dyadic case
($q=\thedim\ge p+2$). For $i_p-1=i_{p+1}=i_{p+2}<i_{p+3}\le\cdots\le i_{\thedim-1}<i_\thedim$,
we obtain (\casedC),
and for $i_{p+2}\ge i_p>i_{p+1}$ we get (\casedD) (a swap).

\begin{sloppypar}
Finally, we may have $\tau$ without a peak, which means
that $\tau=[2^{i_1}|\cdots|2^{i_{\thedim-2}}|2^{i_{\thedim-1}}|2^{i_{\thedim-1}}]$,
$i_1\le \cdots\le i_{\thedim-1}$ (see case (b) of the definition of~$\Vbc$).
Here $\partial_0$ and $\partial_{\thedim-1}$ bring no $\tau'$
(since $\Vbc(\partial_0\tau)=\Vbc(\partial_{\thedim-1}\tau)=\tau$).
For $1\le k\le \thedim-2$ and $i_j<i_{j+1}$, we get a $\tau'$
by swapping $2^{i_j}$ and $2^{i_{j+1}}$ as in (\casedB).
For $j =\thedim$, $\partial_\thedim$ drops the last component,
and if $i_{\thedim-2}<i_{\thedim-1}$, we get a $\tau'$ by splitting
the last component as in~(\casedH).
\end{sloppypar}
\end{proof}

\heading{Acyclicity. } Given Lemma~\ref{l:taucases},
admissibility of $\Vbc$ can be proved quickly.
Here we will check only acyclicity of the
$V\!\partial$-graph, since the non-existence of infinite
paths will be a side-product of the proof of polynomial
boundedness below.

\begin{lemma}\label{l:acyc}
The $V\!\partial$-graph contains no directed cycle.
\end{lemma}

\begin{proof}
If $\tau'=\Vbc(\partial_j\tau)$ is obtained
from $\tau$ as in Lemma~\ref{l:taucases}, then for
$\tau$ not fully dyadic, one of the following can happen:
\begin{enumerate}
\item $q(\tau')>q(\tau)$, i.e., the length of the dyadic part
increases. This is always the case in (\caseEa), (\caseF),
(\caseG), and (\caseH),
and it may also happen in (\caseA) and (\caseC).
\item $q(\tau')=q(\tau)$ and the breakpoint value
decreases. This happens in (\caseA) and (\caseC) (unless $q$ drops)
and also in (\caseE). The latter is not entirely obvious,
since we need to check that $\ltrim(2^{i_{p+1}}+b)<b$,
but this holds since
$\ltrim(2^{i_{p+1}}+b)\le
2^{i_{p+1}}+b-2^{i_p}$,
and $2^{i_p}>2^{i_{p+1}}$.
\item $q(\tau')=q(\tau)$, the breakpoint value stays the same,
and the dyadic part becomes lexicographically larger.
This happens in (\caseB) and (\caseD), since the swaps move a larger component
forward.
\end{enumerate}

If $\tau$ is fully dyadic, then so is $\tau'$, and
either the sum of components of $\tau'$  is smaller
than that of $\tau$ (cases (\casedA) and (\casedH)),
or the sums of components are equal and $\tau'$
is lexicographically larger than $\tau$
(cases (\casedB), (\casedC), and (\casedD)).

This implies that there can be no directed cycle.
\end{proof}

\medskip

We remark that an alternative proof of Lemma~\ref{l:acyc}
can go along the following lines: If $\tau=[a_1|\cdots|a_\thedim]$
is not fully dyadic,
then it can be shown that \emph{either} $\ones(\tau')<\ones(\tau)$,
where $\ones(\tau)$ is the total number of 1's in $a_1,\ldots,a_\thedim$
written in binary, \emph{or} $\ones(\tau')=\ones(\tau)$
and the sequence $(i_1,\ldots,i_p)$ is lexicographically (strictly) larger
than $(i'_1,\ldots,i'_{p'})$, where $2^{i_1}|\cdots|2^{i_p}$
is the dyadic nondecreasing part of $\tau$,
and similarly for $2^{i'_1}|\cdots|2^{i'_{p'}}$ and~$\tau'$.

\heading{Polynomial boundedness. }
Condition (PBV1), polynomial computability of the vector field,
is clearly satisfied for $\Vbc$, and so we need to check (PVB2); i.e.,
we need a polynomial bound on the total encoding size of all simplices
reachable from a given simplex $\sigma$. Obviously, we can
focus only on target simplices: it suffices to provide,
for every target simplex $\ttau$, a polynomial bound on
$\sum_{\tau\in\treach(\ttau)}\size(\tau)$ in terms of $\size(\ttau)$.

Moreover, it is easy to see that neither the application of $\Vbc$
nor the face operators $\partial_i$ can increase the sum
of the components of the simplex. Thus, $\size(\tau)\le
\size(\ttau)$ for every $\tau\in\treach(\ttau)$,
and it is enough to bound the \emph{number} of simplices
in $\treach(\ttau)$.

Thus, let us fix a target simplex $\ttau$ and set
$n:=\size(\ttau)$. Our goal is a polynomial bound,
in terms of $n$, on $|\treach(\ttau)|$.

First we observe that fully dyadic simplices are easily
accounted for. Indeed, a fully dyadic  simplex
$[2^{i_1}|\cdots|2^{i_\thedim}]\in \treach(\ttau)$ is specified
by $i_1,\ldots,i_\thedim\in\{0,1,\ldots,n-1\}$, and so
there are at most $n^\thedim$ such simplices.

So we consider only the $\tau\in\treach(\ttau)$ that are not
fully dyadic. Let us write $\ttau=[\tilde a_1|\cdots|\tilde a_\thedim]$
and $\tau=[2^{i_1}|\cdots|2^{i_q}|b|a_{q+2}|\cdots a_\thedim]$,
where $q=q(\tau)$ is the length
of the dyadic part and $b$ is the breakpoint value.

We would like to show that with $\ttau$ fixed, there are
only polynomially many possibilities for $\tau$.
First, as was noted above, the number of choices
for the dyadic part of $\tau$ is polynomially bounded.

Second, it turns out that all of the right part
of $\tau$ is inherited from $\ttau$, i.e.,
$a_i=\tilde a_i$ for all $i\ge q+2$. This ``stability of
the right part'' is not hard to prove
inductively using Lemma~\ref{l:taucases}, and it will
be the first part of the key lemma below.

Thus, the last thing to do is showing that there are
only polynomially many possibilities for the breakpoint
value $b$ of $\tau$, and this is the most tricky part of the proof.
We will distinguish two cases: if $b=\tilde a_{q+1}$, i.e.,
$b$ is ``inherited'' from $\ttau$, then we call $\tau$ a \emph{raw}
simplex, and otherwise, $\tau$ is \emph{processed}.

The following lemma shows that if $\tau$ is processed, then
its breakpoint value belongs to a certain inductively defined
set, which is of polynomial size. In order that the proof
goes through, we need to strengthen the inductive hypothesis:
namely, we need that for a processed $\tau$, the breakpoint
value is smaller than the maximum entry of the dyadic part.
This will play a role only in a single case among those
in Lemma~\ref{l:taucases}, namely (\caseE); while all the other cases
are natural and straightforward, (\caseE) seems to work only by a small
miracle.

\begin{lemma}[Key lemma]\label{l:ttau}
 Let $\tau\in \treach(\ttau)$ be as above.
Then $a_i=\tilde a_i$ for all $i\ge q+2$, i.e., the
right part of $\tau$ coincides with the corresponding segment
of $\ttau$. Moreover, if $\tau$ is processed,
then $b<\max(2^{i_1},\ldots,2^{i_q})$,
and $b\in B_{q+1}$, where the sets $B_1,\ldots,B_\thedim$ are defined
inductively as follows:
\begin{itemize}
\item $B_1=\ltrims(\tilde a_1)$, where, for a positive integer $a$,
we define $\ltrims(a)=\emptyset$ if $a$ is a power of~$2$,
and $\ltrims(a)=\{\ltrim(a)\}\cup \ltrims(\ltrim(a))$ otherwise.
\item $
B_{j+1}=\ltrims\Bigl(\{\tilde a_{j+1},\tilde a_{j}+\tilde a_{j+1}\}
\cup  \{2^i+\tilde a_{j+1}:0\le i\le n-1\}
\cup \{b+\tilde a_{j+1}: b\in B_{j}\}\Bigr),
$
where we extend $\ltrims(.)$ to sets by $\ltrims(A):=
\bigcup_{a\in A}\ltrims(a)$.
\end{itemize}
\end{lemma}

\begin{proof} It suffices to prove that if $\tau$ is as claimed
in the lemma, then $\tau'=\Vbc(\partial_j\tau)$ as in Lemma~\ref{l:taucases}
has this form as well (moreover, we may
assume that $\tau'$ is not fully dyadic).
We need to consider the cases (\caseA)--(\caseH)
in Lemma~\ref{l:taucases}, but we can right away settle (\caseH),
where $\tau'$ is fully dyadic, as well as (\caseB) and (\caseD), which only
permute the dyadic part.
This leaves us with cases (\caseA), (\caseC), (\caseE),
(\caseEa), (\caseF), and~(\caseG).

First let $\tau$ be raw, with $b=\tilde a_{q+1}$.
In cases (\caseA) and (\caseC)
$\tau'$ contains $\lpow(\tilde a_{q+2})$ followed by
$b':=\ltrim(\tilde a_{q+2})$, at the $(q+1)$st position.
If $b'$ is a power of $2$, then $\tau'$ is raw, and otherwise,
we have $b'\in B_{q+1}$ and $b'<\lpow(\tilde a_{q+1})$;
the latter is the required entry larger than $b'$ in the dyadic part.
Hence $\tau'$ is a processed simplex as claimed in the lemma.

In (\caseE) and (\caseF), we have a situation similar to the one just discussed,
except that $b'=\tilde a_{q+1}+2^i$ for some $i< n$ in (\caseE),
and $b'=\tilde a_{q+1}+\tilde a_{q+2}$ in (\caseF).
Moreover, in (\caseE), $b'$ is at position $q+1$,
while in (\caseF) it is at position $q+2$. Again we find that $\tau'$
is a processed simplex of the claimed form.
In cases (\caseEa) and (\caseG), we either get $\tau'$ fully dyadic, or
the breakpoint value of $\tau'$ is $\ltrim(\tilde a_{q'+1})$ for some
$q'\ge q+1$, preceded by $\lpow(\tilde a_{q'+1})$.
Then $\tau'$ is a processed simplex as in the lemma as well,
and the discussion
of a raw $\tau$ is finished.

Now let $\tau$ be processed, with $b\in B_{q+1}$,
$b<\max(2^{i_1},\ldots,2^{i_q})$.
In cases (\caseA) and (\caseC) $\tau'$ may be raw,
which is fine, or processed with breakpoint value $\ltrim(b)$,
which lies in $B_{q+1}$, since $B_{q+1}$ is closed under $\ltrim(.)$.

Case (\caseE) is, in a sense, the most sophisticated, and it is here
where the inductive hypothesis $b<\max(2^{i_1},\ldots,2^{i_q})$
is crucial. In the setting of (\caseE), $2^{i_p}$ is the maximum
of the dyadic part of $\tau$, and so $2^{i_p}>b$. Let $b'=b+2^{i_{p+1}}$,
where $2^{i_{p+1}}<2^{i_p}$; by the conditions in case (\caseE),
we have $b'>2^{i_p}$.

We claim that $\ltrim(b')\in \ltrims(b)$ (this will show that $b'\in B_{q+1}$
and thus $\tau'$ is as required). To check this, let us write, for brevity,
$u=i_p$ and $v=i_{p+1}$, and let $\beta_{u-1}\beta_{u-2}\cdots\beta_0$ be
the binary notation for $b$, i.e., $b=\sum_{i=0}^{u-1}\beta_i 2^i$,
$\beta_i\in\{0,1\}$. Since $2^u-2^v<b<2^u$,
we have $\beta_{u-1}=\cdots=\beta_v=1$.
Then $b'$ in binary is $1000\cdots0\beta_{v-1}\beta_{v-2}\cdots
\beta_0$, and so $\ltrim(b')$ can be obtained from $b$
by iterating $\ltrim(.)$. Thus, $b'\in B_{q+1}$ indeed.

The consideration in cases (\caseEa) and (\caseG) is
the same
as the one for $\ttau$ raw.

The last case to consider is (\caseF). Here the dyadic
part of $\tau'$ is longer than that of $\tau$.
By induction, we have $b\in B_{q+1}$, and so $\ltrim(b+\tilde a_{q+2})\in B_{q+2}$
by the definition of $B_{q+2}$ (or it is a power of $2$, in which
case $\tau'$ is raw). As in the previous
case, the entry $\lpow(b+\tilde a_{q+2})$
supplies the power of $2$ greater than $\ltrim(b+\tilde a_{q+2})$,
as required for the induction.
The lemma is proved.
\end{proof}

\begin{corol}\label{c:pbou}
For $\ttau$ as in Lemma~\ref{l:ttau}, we have
$|\treach(\ttau)|= O(n^{2\thedim})$, with the implicit constant
depending on~$\thedim $.
\end{corol}

\begin{proof} For each $\tau\in \treach(\ttau)$, we have
at most $n^\thedim$ choices for the dyadic part (which includes
fixing $q$, the length of the dyadic part). A raw $\tau$
is already determined by $\ttau$ and by the dyadic part, while for
$\tau$ processed, we also need to specify~$b$.

The definition of $B_j$ gives $|B_1|\le n$ and $|B_{j+1}|\le 3n+n^2+n|B_j|$,
which yields $|B_j|=O(n^j)$, and the corollary follows.
\end{proof}

\heading{Remark. } A more careful (and more complicated)
analysis should probably give $O(n^\thedim)$ instead of
$O(n^{2\thedim})$ in Corollary~\ref{c:pbou}. However, as we
will now indicate, our vector field is not much better;
there can indeed be about $n^\thedim$ reachable simplices
in~$\treach(\ttau)$.

To see this, let us take $n$ that is
an integer multiple of $\thedim^2$, i.e., $n=\thedim^2\ell$, and let
us consider a source simplex
$\tsigma=[\tilde a_1|\cdots|\tilde a_\thedim]$, where
$\tilde a_i:= (2^\ell-1)2^{(i-1)\ell}$, $i=1,2,\ldots,k$.
Put differently, if we think of the binary encoding of
each $\tilde a_i$ as consisting of $\thedim$ blocks of $\ell$ bits each
(thus, $\tilde a_i$ has at most $n/\thedim$ bits and $\size(\tsigma)\le n$),
then $\tilde a_i$ has $1$'s in the $i$th block and $0$'s elsewhere.
It can be shown that each simplex $\sigma=[a_1|\cdots|a_\thedim]$,
where $a_i$ has exactly one 1 in the $i$th block and $0$'s everywhere
else,  belongs to $\reach(\tsigma)$. Since for each $i$, the position of the
single $1$ in $a_i$ can be chosen in $\ell$ ways,
we have $|\reach(\tsigma)|\ge\ell^\thedim=(n/\thedim^2)^\thedim$.

It would be interesting to see
if one could reach a significantly better bound with a different
vector field, or if there is perhaps a good lower bound valid for every
vector field.


%

 \subsection*{Acknowledgments}

 We would like to thank
Martin \v{C}adek, Luk\'a\v{s}
Vok\v{r}\'{\i}nek, and Uli Wagner for useful discussions and
ongoing collaboration. Moreover, we thank Uli Wagner  and Martin \v{C}adek
for insightful comments on a preliminary version of the manuscript.

\bibliographystyle{alpha}
\bibliography{Postnikov,simplicHomotopy}

\end{document}